\documentclass{article}
\usepackage[english]{babel}
\usepackage{tikz}        
\usepackage{xcolor}
\usepackage{amsmath,amsthm,amssymb,mathrsfs,pifont,cancel,paralist}%
\usepackage[utf8]{inputenc} 
\usepackage{empheq} 
\usepackage{nameref} 
\usepackage[notcite,notref,final]{showkeys} 
\usepackage{ulem} 
\usepackage{authblk}
\usepackage{natbib}

\newtheorem{prop}{Proposition} %

\usepackage[bookmarks,
colorlinks=true, linkcolor=black,
anchorcolor=black, citecolor=black, filecolor=black,
menucolor=black, pagecolor=black, urlcolor=black,
pdftitle={One dimensionale actin growth model}]{hyperref}

\newcommand{\brac}[1]{\left( {#1} \right)}

\newcommand{\red}[1]{\textcolor{red}{#1}}

\definecolor{forestgreen}{cmyk}{0.91,0,0.88,0.12}
\definecolor{brown}{cmyk}{0,0.81,1,0.60}
\definecolor{fuchsia}{cmyk}{0.47,0.91,0,0.08}

\newcommand{\Dert}[2]{\frac{d #1}{d #2}} 
\newcommand{\Derp}[2]{\frac{\partial #1}{\partial #2}} 
\newcommand{\iDerp}[2]{{\partial #1}/{\partial #2}} 

\newfont{\tenbfit}{cmmib10}%
\newfont{\svnbfit}{cmmib8}%
\newfont{\tenbfsl}{cmbxti10}

\newfont{\mmit} {cmmi10}%
\newcommand{\lj}{\mbox{$[\kern-0.1478125em[$}}
\newcommand{\rj}{\mbox{$]\kern-0.1478125em]$}}
\newcommand{\la}{\mbox{$\langle\kern-0.2325em\langle$}}
\newcommand{\ra}{\mbox{$\rangle\kern-0.2325em\rangle$}}
\newcommand{\Blj}{\mbox{$\Big[\kern-0.25em\Big[$}}
\newcommand{\Brj}{\mbox{$\Big]\kern-0.25em\Big]$}}

\begin{document}

\title{Treadmilling stability \\ of a one-dimensional actin growth model}
\date{July 10th 2019}
\author{Rohan Abeyaratne\thanks{rohan@mit.edu}}
\affil{Department of Mechanical Engineering, Massachusetts Institute of Technology, Cambridge, MA, USA}
\author{Eric Puntel\thanks{Corresponding Author: eric.puntel@uniud.it}}
\affil{Dipartimento Politecnico di Ingegneria e Architettura, Università di Udine, via del Cotonificio 114, Udine I-33100, Italy}
\author{Giuseppe Tomassetti\thanks{giuseppe.tomassetti@uniroma3.it}}
\affil{Dipartimento di Ingegneria, Università degli Studi Roma Tre, via Volterra 62, Roma I-00146, Italy}

\maketitle
\begin{abstract}
Actin growth is a fundamental biophysical process and it is, at the same time, a prototypical example of diffusion-mediated surface growth.
We formulate a coupled chemo-mechanical, one-dimensional growth model encompassing both material accretion and ablation. A solid bar composed of bound actin monomers is fixed at one end and connected to an elastic device at the other. This spring-like device could, for example, be the cantilever tip of an AFM.  The compressive force applied by the spring on the bar increases as the solid grows and affects the rate of growth. The mechanical behaviour of the bar, the diffusion of free actin monomers in a surrounding  solvent and the kinetic growth laws at the accreting/ablating ends are accounted for. The constitutive response of actin is  modeled by a convex but otherwise arbitrary  elastic strain energy density function. Treadmilling solutions, characterized by a constant length of the continuously evolving body, are investigated. Existence and stability results are condensed in the form of simple formulas and their physical implications are discussed.
\end{abstract}

\section{\label{sec:intro}Introduction}
It is well known that growth in living systems is not only promoted by biological and chemical signals but is also affected by mechanical stimuli \citep{Goriely:2017}.

Modelling growth, intended as variation of mass, poses a number of challenges in mechanics which are still being actively investigated. Among them is surface growth or accretion which, following the work of Skalak and others \citep{SkalakDasgupta1982,SkalakFarrow1997}, requires one to define and track in time an ever changing, usually stress-free, reference configuration, i.e. collection of material points. The phenomenon of accretion of a solid on its boundary, occurs in several contexts of physical, technological, and biological interest. One of the most common examples of surface growth is the solidification of water at the ice-water interface near the freezing temperature; other examples include technological processes such as chemical vapor deposition, 3D printing and layered building \citep{BacigalupoGambarotta:2012,ZurloTruskinovsky2017,ZurloTruskinovsky2018,TruskinovskyZurlo2019}; in biology, the growth of hard tissues like bones and teeth \citep{CiarlettaPreziosi2013,GanghofferGoda2018}. When surface growth occurs at an interior surface it generates stress, since each new layer of solid material that forms must push away the layers deposited previously.

A second delicate issue regards the prescription of a growth law. One may simply assume that as given. Conversely, growth speed could be obtained as a result of mechanical and biochemical local conditions. These in turn may be expressed by a suitable kinetic law once the thermodynamical force driving growth has been consistently defined \citep{AbeyaratneKnowles1990,TomassettiCohen:2016}.

Third, one may also describe the transport of the free particles that provide the material constituents for growth. In this way essential features of growth may emerge from the balance of coupled mechanical and biochemical responses.

In this work we formulate and analyze a one-dimensional model featuring the three aforementioned characteristics, albeit in a simplified manner. We consider an elastic bar fixed at one end and connected to an elastic device at the other. This spring-like elastic device could for example be the cantilever tip of an atomic force microscope as in the experiments described in \citep{ParekhChaudhuri2005,ChaudhuriParekh2007,BielingLi2016}.  The bar can grow by attaching or detaching its constituting particles (``monomers''), at either end. The diffusion of free particles in a surrounding or permeating solvent and the kinetic condition for growth are accounted for. The first objective of this study is to investigate a basic reference template of chemo-mechanical growth which allows one to discuss more easily modelling choices, notions and solutions.

The second motivation for this study is provided by a specific biological example, namely the growth of actin filaments. Actin in its polymerized network-forming state is an essential constituent of the cytoskeleton and is involved in cell contraction, division and motility. It is intensely studied in the bio-physical literature. See e.g.\ \cite{ProstJuelicher2015} for a review on the physics of active gels like actin, the Ph.D. thesis of \cite{Zimmermann2014} for a review of quantitative models of actin-based motility, and \cite{BindschadlerOsborn2004} and \cite{CardamoneLaio2011} for just two of the many examples of different biophysical and computational models of the properties of actin networks. Pertinent to this study, but not including mechanical aspects, is a one-dimensional mathematical model of actin polymerization kinetics by \cite{Edelstein-KeshetErmentrout2000}.

Of particular interest are the experimental studies described in \cite{ParekhChaudhuri2005,ChaudhuriParekh2007} and \cite{BielingLi2016} that involve an experimental setup similar to the one considered here, where an actin network is grown between the cantilever tip of an atomic force microscope (AFM) and a fixed surface below it thus realizing a bar-like structure fixed at one end and restrained by an elastic device at the other. Among other things, these experiments consistently suggest that the actin network adapts to higher values of applied compressive force by correspondingly increasing its density and stiffness. This is a feature that is currently not included in our model, but it constitutes a possible refinement for future work. 

Actin filaments exhibit a peculiar growth mode called treadmilling in which the length of the filament in physical space remains constant while accreting (i.e. attaching) actin monomers at one end and ablating (i.e. detaching) them at the other at equal rates \citep[see e.g.][]{Theriot2000}. This energy dissipating state is made possible by the energy provided by the hydrolysis of the ATP (adenosine triphosphate) bound to actin monomers into an ADP (adenosine diphosphate) molecule and a phosphate. Despite its peculiarity, treadmilling may also be seen as a specific instance of a more common biological paradigm by which systems, tissues or organisms continuously substitute their constituents or cells at specific rates even when their overall size is no longer changing.

The present work has two main results. First, under rather general assumptions on the behaviour of the material constituting the bar, we establish conditions for the existence of treadmilling states in terms of simple formulas.

Second, the stability of such solutions is discussed. Herein stability is not addressed energetically in classical structural mechanics terms (i.e. buckling) by considering perturbations perpendicular to the bar but rather dynamically, asking whether perturbations in the direction of the bar of the treadmilling state may cause the bar to abandon indefinitely its stationary length. For the treadmilling case in which ATP-actin accretes at the fixed end of the bar, it is possible to prove global stability under arbitrary initial conditions.


We expect that the present results will be a useful tool in the comparison and interpretation of other more complex models and experiments. For instance, our results on the global stability of a treadmilling solution may provide clarification or further evidence in support of the ``emergence of a universal growth path'' observed in a similar context by \cite{Abi-AklAbeyaratne2019}. 

The discussion on stability of the solution has also been motivated by experiments studying the growth and relative stability of an annulus of actin accreting on the surface of a spherical bead \citep{CameronFooter1999,NoireauxGolsteyn:2000,GuchtPaluch2005}. These experiments were in turn inspired by bacterium Listeria monocytogenes which exploits cytoplasmic actin to form a polymerized tail and move out of the cell membrane and spread \citep{ProstJoanny2008}. Existing numerical and modelling efforts on this subject can be found in \citep{JohnPeyla2008,BuylMikhailov2013}.

A model for a spherical annulus of actin growing on the surface of a sphere was formulated and analyzed by \cite{TomassettiCohen:2016}. Treadmilling solutions were studied there as well but not their stability. We plan to continue the above study and the present one by the analysis of the stability of the treadmilling solutions of that spherical annular system. Other extension of the present work include accounting for the aforementioned dependence of actin density and stiffness on externally applied stress and deriving analytical relationships between growth velocity and external force to be compared with experimental ones.

This paper is structured as follows. Section \ref{sec:model} describes the one-dimensional model including its mechanical, chemical and growth aspects. Section \ref{sec:const} provides the material constitutive description of the bar while the derivation of the driving force is given in Section \ref{sec:kinder}. The system is reduced to a differential algebraic equation in Section \ref{sec:DAE}, a form that is particularly suitable for the  subsequent discussion of the existence, uniqueness and stability of treadmilling solutions that is carried out in Section \ref{sec:TM}. The results are discussed in Section \ref{sec:disc} and concluding remarks are made in Section \ref{sec:end}. 

\section{\label{sec:model}One-dimensional model}

 \subsection{\label{ss:set} Problem setting}
  \begin{figure}[h]
  \begin{center}
  \includegraphics[width=0.75\linewidth]{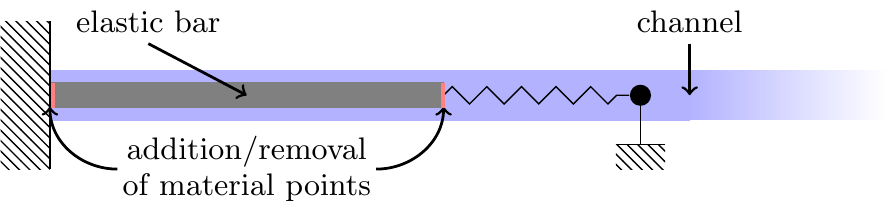}
  \end{center}
  \caption{\footnotesize An elastic bar clamped between a hard and a soft device, immersed in a semi-infinite channel.}
  \label{fig:setup}
  \normalsize
  \end{figure}
  We consider a one-dimensional body, represented by a bar in Figure \ref{fig:setup}, which grows and deforms in a one-dimensional physical space.

  The bar has a natural reference configuration that occupies the segment $(x_0(t),x_1(t))$ and whose generic point is denoted by $x$ . Here and in the following subscripts $0$ and $1$ refer to the left and right end sections of the bar respectively, both in the reference and in the current configurations.

  As represented in Figure \ref{fig:onedim}, the body is mapped into the physical one-dimensional space through the function $y(x,t)$ where it occupies the segment $(y_0(t),y_1(t))$. Here and in what follows the shorthand notation
  \begin{equation}\label{eq:fam}
     f_\alpha(t) = f(x_\alpha(t),t) \qquad \text{with } \alpha = 0,1 \;,
  \end{equation}
  denotes in general the value of a material quantity $f(x,t)$ at the end sections of the bar at time $t$. In particular $y_0(t)$ and $y_1(t)$ simply indicate the position of the end sections of the bar in the current configuration.


  %
  \begin{figure}[h]
  \begin{center}
  \includegraphics[width=0.75\linewidth]{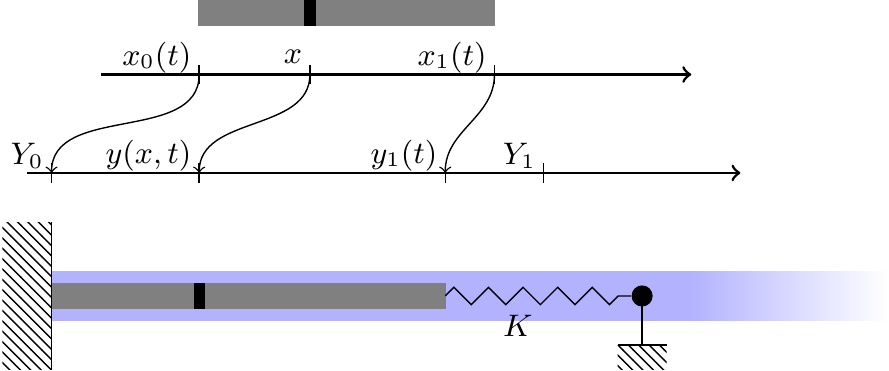}
  \end{center}
  \caption{\footnotesize  Reference (top) and current (bottom) configuration of the elastic bar.}
  \label{fig:onedim}
  \normalsize
  \end{figure}
  In regard to constraints, the terminal side $x_0$ of the bar is attached to the point $Y_0$ in the physical space, so that $y_0=y(x_0(t),t)=Y_0$. Likewise, the terminal side $x_1$ is attached to one end of a linear spring of stiffness $K$. The rest position of this end of the spring is $Y_1$, i.e. the spring is unstretched with zero force when this end is at $Y_1$. Its other end is fixed. It is worth noting that while the left-hand end of the bar is always located at $Y_0$ in physical space, the right-hand end is located at $Y_1$ only when the spring force vanishes.

  The bar is made of ``material units'', hereafter referred to as \emph{monomers}, which are in a bound, polymerized state. The same monomers in a free, unbound state are in solution in the solvent which fills a one dimensional infinite channel, depicted with a blue-shaded rectangle in Figure \ref{fig:setup}. Free monomers flow in the interval $(y_0(t),y_1(t))=(Y_0,y_1(t))$ according to Fick's law. We can think of them as either flowing only through the bar or flowing as well through the portion of the channel not occupied by the bar. They can freely cross the point $y_1$, where the body is in contact with a reservoir of monomers, but cannot flow past the left support $y_0 = Y_0$ which is assumed to be impermeable.
  The chemical potential $\mu_1$ of free monomers at $y_1(t)$ is held fixed and equal to $M_1$ , and there is an infinite supply of monomers at $y_1$.

  Finally, under suitable growth conditions to be later specified, free monomers may accrete, i.e. attach, at either end of the bar and conversely, bound monomers occupying the end positions $x_0$ and $x_1$ of the bar may ablate, i.e. detach, and return to their free state. When accretion or ablation occurs, the referential points $x_0$ and $x_1$, and hence the reference length of the bar, can change. Specifically, at the left-hand end $x_0$, accretion occurs when $\dot{x}_0 < 0$ while ablation occurs when $\dot{x}_0 >0$. Similarly accretion and ablation at the right-hand end $x_1$ correspond to $\dot{x}_1 >0$ and $\dot{x}_1 <0$ respectively.

  We now specify the equations governing the system just described.

  \subsection{\label{ss:mech}Mechanics}
  We require the deformation mapping $y(x,t): x \rightarrow y$ to be one to one by prescribing that the stretch $\lambda = y' = \iDerp{y}{x}$ be positive:
  \begin{equation*}
    \lambda(x,t) = y'(x,t) > 0 \;.
  \end{equation*}
  Here and in the following we use the prime to denote the derivative with respect to a variable other than time. That is $f' = \iDerp{f}{\bullet}$ with $f=f(\bullet)$ or $f=f(\bullet,t)$ with $t$ indicating time and $\bullet \neq t$. The dot is used, as customary, to indicate partial derivative with respect to time, i.e. $\dot{f} = \iDerp{f}{t}$.

  We assume the material to be hyperelastic and characterized by a convex strain energy density $W(\lambda)$ from which we can compute the axial force $\sigma$ in the bar as
  \begin{equation*}
    \sigma(x,t) = W'(\lambda(x,t)) \;.
  \end{equation*}
  A number of additional assumptions on the strain energy density are detailed in Section \ref{ss:W}. Given that $W$ does not depend explicitly on $x$, the material constituting the bar is taken to be homogeneous.

  The mechanical model for the bar is summarized in the following set of equations:

  \begin{subequations}\label{eq:mech}
    \begin{empheq}[left=\empheqlbrace]{align}
      &\Derp{\sigma}{x}=0& &\text{in }(x_0(t),x_1(t)), \label{eq:mech1}\\
      &\sigma= W'(\lambda), \quad \lambda = y'& &\text{in }(x_0(t),x_1(t)), \label{eq:mech2}\\
      &y_0(t) = y(x_0(t),t)=Y_0& &\text{in }x_0(t), \label{eq:mech3}\\
      &\sigma_1(t) +K(y_1(t)-Y_1) = 0& &\text{in }x_1(t)\;.  \label{eq:mech4}
    \end{empheq}
  \end{subequations}

  Equation \eqref{eq:mech1} represents local equilibrium in the reference configuration and implies that the axial force is constant in $x$. In eq.\ \eqref{eq:mech2} we state again the constitutive law and the definition of the stretch $\lambda$. The force $\sigma$ being constant in $x$, it follows that $\lambda$ is constant in $x$ and that $y$ is linear in $x$.

  The boundary condition prescribing that the leftmost section of the bar is fixed in $y=Y_0$ is expressed in eq.\ \eqref{eq:mech3}. The axial force $\sigma_1(t) = \sigma(x_1(t),t)$ in the bar at $x=x_1(t)$ is prescribed by the force in the spring of stiffness $K$ in eq. \eqref{eq:mech4}. Since $\sigma$ is independent of $x$, \eqref{eq:mech4} actually prescribes the value of the axial force in the whole bar.

  \subsection{\label{ss:diff}Diffusion}

  The following system
  \begin{subequations}\label{eq:diff}
    \begin{empheq}[left=\empheqlbrace]{align}
      &\Derp{h}{y}=0& &\text{in }(Y_0,y_1(t)), \label{eq:diff1}\\
      &h+m\mu'=0& &\text{in }(Y_0,y_1(t)), \label{eq:diff2}\\
      &h(y_0(t),t) = h(Y_0,t)=\varrho \dot x_0(t)& &\text{in }Y_0, \label{eq:diff3}\\
      &\mu(y_1(t),t)= M_1 & &\text{in }y_1(t)\;. \label{eq:diff4}
    \end{empheq}
  \end{subequations}
  %
  governs the flux of free monomers in the solvent. Here $h(y,t)$ is the monomer flux in the positive $y$ direction, $\mu(y,t)$ is the associated chemical potential and $m$ is the mobility.
  The first equation \eqref{eq:diff1} expresses the conservation of mass. In it we have omitted a term $\iDerp{h}{t}$ by assuming that diffusion is much faster than growth. Flux has the dimension of moles per unit time. The second equation \eqref{eq:diff2} represents Fick's law.
  The third equation \eqref{eq:diff3}, a boundary balance of mass, states that the flux of monomers at the impermeable wall is equal to the amount of monomers that detach from the left endpoint of the bar per unit time, which in turn is proportional to the ablation velocity $\dot x_0$ through a constant $\varrho$. We think of $\varrho$ as the number of moles of bound actin monomer per unit length in the reference configuration. The fourth equation  \eqref{eq:diff4} expresses the condition of chemical equilibrium at $y_1$ by equating $\mu_1 = \mu(y_1(t), t)$ to the chemical potential $M_1$ of the monomers in the semi-infinite monomer pool to the right of $y_1$. Note that $\mu_0 = \mu(y_0,t)$ is as yet unknown and to be determined.  

  \subsection{\label{ss:accr}Accretion}

  As anticipated in the \nameref{sec:intro}, a key ingredient of this model is the growth law governing the evolution of the referential configuration of the bar. We assume a simple, linear kinetic law of the form
  \begin{equation}\label{eq:kinshort}
    B_\alpha V_\alpha=F_\alpha \qquad \text{with } \alpha = 0,1
  \end{equation}
  In eq.\ \eqref{eq:kinshort}, $\alpha =0,1$ refer to the ends of the bar, $V_\alpha$ is the accretion velocity, $B_\alpha$ is a positive kinetic coefficient and $F_\alpha$ is the thermodynamical force driving accretion. Note that
  \begin{equation}\label{eq:Vxdot}
    V_0 = - \dot{x}_0 \quad \text{and} \quad V_1 = \dot{x}_1 \,.
  \end{equation}
  Realistically \eqref{eq:kinshort} is most suitable for small deviations from thermodynamic equilibrium. The expression for the driving force $F_\alpha$, derived afterwards in Section \ref{sec:kinder}, is
  \begin{equation}\label{eq:Fa}
    F_\alpha=\varrho(\mu_\alpha-M_{B,\alpha})+W^*(\sigma_\alpha),
  \end{equation}
  where $\sigma_\alpha$ and $\mu_\alpha$ are \emph{material} descriptions of the fields $\mu$, $\sigma$ evaluated at $x_\alpha$ following the notation introduced in \eqref{eq:fam}. Parameter $M_{B,\alpha}$ is a material constant interpreted as the chemical potential of \emph{bound} monomers at $x_\alpha$ and $W^*(\sigma)$ is the \emph{complementary} strain energy density whose definition and properties are given in Section \ref{ss:W*}. In particular we will see there that a tensile force $\sigma > 0$ corresponds to positive $W^*(\sigma)$ thus promoting growth according to \eqref{eq:kinshort}-\eqref{eq:Fa}. This is consistent with the layman's notion of stress induced growth popularized by images of abnormal growth of earlobes, necks and other body parts subject to sustained tension, especially observed in some indigenous tribes \citep[see e.g.\ ][Chapter 2.1]{Goriely:2017}.

  Motivated by the behaviour of actin filaments, see e.g.\ \cite{Theriot2000} or \citet[][panel 16-2]{AlbertsJohnson2015}, we admit two distinct values $M_{B,0}$ and $M_{B,1}$ for the chemical potential of monomers in the bound state. Actin monomers are bound to ATP (Adenosine TriPhosphate) when they first polymerize, i.e. accrete, but after some time a hydrolysis reaction ensues by which the ATP releases a phosphate group and the polymerized actin monomer is now tied to an ADP (Adenosine DiPhosphate) molecule. The hydrolysis reaction releases energy, part of which remains stored in the polymerized actin. Therefore ADP-actin is at a higher energy level, i.e. chemical potential, than ATP-actin. Due to differences in the properties of opposite ends of actin filaments, one end may be occupied by a lower-energy ATP-actin monomer and the other by a higher energy ADP bound actin monomer. Hence the distinction between values of the chemical potential of polymerized actin $M_{B,0}$ and $M_{B,1}$  at the two ends of the bar.
  %
  %

  As noted previously, a positive accretion velocity corresponds to a negative rate $\dot{x}_0$ and to a positive rate $\dot{x}_1$ whence $V_0 = - \dot{x}_0$ and $V_1 = \dot{x}_1$, and recalling that we are using the notation of equation \eqref{eq:fam}, the specialization of \eqref{eq:kinshort} and \eqref{eq:Fa} to the two ends of the bar can be written as
  \begin{subequations}\label{eq:accr}
    \begin{empheq}[left=\empheqlbrace]{align}
      -B_0\dot x_0(t)&=\varrho(\mu_0-M_{B,0})+W^*(\sigma_0(t))\;, \label{eq:accr1}\\
      B_1\dot x_1(t)&=\varrho(\mu_1-M_{B,1})+W^*(\sigma_1(t))\;. \label{eq:accr2}
    \end{empheq}
  \end{subequations}
  Despite the simplicity of the one-dimensional model, the above equations close the feedback loop between stress and growth. On the one hand, the presence of the spring in \eqref{eq:mech} allows growth to affect stress, while on the other hand, growth rates in \eqref{eq:accr} are influenced by stress.

  The evolution equations \eqref{eq:accr} also provide closure for the boundary-value problem \eqref{eq:mech} and \eqref{eq:diff}. In fact, the solution of \eqref{eq:mech}--\eqref{eq:diff} depends only on the instantaneous values of $x_0(t)$, $x_1(t)$ and of the rate $\dot{x}_0(t)$. This means, in particular, that the right-hand sides of the equations \eqref{eq:accr} ultimately depend only on $x_0(t)$, $x_1(t)$ and $\dot{x}_0(t)$. We therefore conclude that the combination of \eqref{eq:mech}, \eqref{eq:diff}, and \eqref{eq:accr} is equivalent to a first-order system in the unknowns $x_0(t)$ and $x_1(t)$. As such, this system must be complemented by suitable initial conditions

\section{\label{sec:const}Constitutive behaviour}
We assume the bar to be made of a homogeneous, hyperelastic material and we define its constitutive behaviour through the strain energy density function $W(\lambda)$. As seen in equation \eqref{eq:mech2}, the force $\sigma$ is given by $W'(\lambda)$ while $W''(\lambda)$ represents the tangent stiffness.

  \subsection{Strain energy density} \label{ss:W}

  A specific expression for $W(\lambda)$ is not prescribed. Instead, we merely assume that the strain energy density has the following characteristics:
  \begin{subequations}\label{eq:Wass}
    \begin{empheq}[left=\empheqlbrace]{align}
      &W(1) = 0 \label{eq:Wass1} \\
      &W'(1) = 0 \label{eq:Wass2} \\
      &W(\lambda) \to +\infty & & {\rm as} \ \lambda \to 0^+  \label{eq:Wass3}\\
      &W(\lambda) \to +\infty & &{\rm as} \ \lambda \to +\infty  \label{eq:Wass4}\\
      &W'(\lambda) \to +\infty & &{\rm as} \ \lambda \to +\infty  \label{eq:Wass5}\\
      &W''(\lambda) > 0  & &\forall \lambda > 0 \label{eq:Wass6}
    \end{empheq}
  \end{subequations}
  discussed in the following.

  The strain energy density is defined but for an arbitrary constant which is conveniently set in \eqref{eq:Wass1} by assigning zero energy to the undeformed state in which the stretch $\lambda$ is equal to $1$. In this latter case the force $\sigma$ is also zero according to eq.\ \eqref{eq:Wass2}. Equations \eqref{eq:Wass3} and \eqref{eq:Wass4} express the requirement that infinite strain energy is necessary to, respectively, infinitely compress and infinitely extend the material. Condition \eqref{eq:Wass6} on $W(\lambda)$ enforces convexity of the strain energy density which in turn implies that the stress is a monotonic function of the stretch, and the tangent stiffness is positive everywhere, i.e. there are no stress softening branches under increasing stretch. Note that $\sigma>0$ for $\lambda>1$ and $\sigma<0$ for $0 < \lambda<1$. Assuming sufficient regularity of $W(\lambda)$, equations \eqref{eq:Wass3} and \eqref{eq:Wass6} can be used to prove that, when the stretch tends to zero, the force tends to infinity, that is $\sigma \to -\infty$ when $\lambda \to 0^+$. It is easy to see that conditions \eqref{eq:Wass4} and \eqref{eq:Wass6} are not sufficient to obtain an analogous result for the case in which the stretch tends to infinity. The condition that $\sigma \to +\infty$ when $\lambda \to +\infty$ is therefore explicitly given in \eqref{eq:Wass5}. We note in passing that the set of assumptions \eqref{eq:Wass} is introduced in a constructive way and is not minimal since \eqref{eq:Wass4} follows from \eqref{eq:Wass5} and \eqref{eq:Wass6}.

  From the properties of $W(\lambda)$ ensue those of the force $\sigma$. Let
  \begin{equation}\label{eq:shat}
    \widehat{\sigma}(\lambda) := W'(\lambda) \quad \rm{,} \quad
    \widehat{\sigma}(\lambda) : \mathbb{R}^+ \longrightarrow \mathbb{R} \;.
  \end{equation}
  Then from \eqref{eq:Wass6} we know that $\widehat{\sigma}$ is monotonically increasing and, taking into account \eqref{eq:Wass3}-\eqref{eq:Wass5} as well, that it spans the entire real line.

  The force $\widehat{\sigma}$ as a function of the stretch is therefore invertible and the function
  \begin{equation}\label{eq:lhat}
    \widehat{\lambda}(\sigma) : \mathbb{R} \longrightarrow \mathbb{R}^+ \;, \quad \widehat{\lambda}(\sigma) \text{ such that }  W'(\widehat{\lambda}) = \sigma
  \end{equation}
  is uniquely defined. It can be easily seen that $\widehat{\lambda}(\sigma)$ is monotonically increasing
  \begin{equation} \label{eq:lhat'}
    \widehat{\lambda}'(\sigma) = \frac{1}{W''(\lambda)} > 0
  \end{equation}
  and that it possesses the following properties:
  \begin{equation} \label{eq:lhatprop}
    \widehat{\lambda}(\sigma) \to 0^+ \ {\rm as} \ \sigma \to -\infty, \qquad \widehat{\lambda}(0) = 1, \qquad
    \widehat{\lambda}(\sigma) \to  +\infty \ {\rm as} \ \sigma \to +\infty.
  \end{equation}

  \subsection{Complementary strain energy density} \label{ss:W*}

  The complementary strain energy density $W^*(\sigma)$ is defined as
  \begin{equation}\label{eq:W*def}
    W^*(\sigma) = \sigma \widehat{\lambda}(\sigma) - W(\widehat{\lambda}(\sigma)) \;\red{.}
  \end{equation}
  $W^*(\sigma)$ is the Legendre transform of $W(\lambda)$. It has the properties
  \begin{subequations}\label{eq:W*}
    \begin{empheq}[left=\empheqlbrace]{align}
      &W^*\, '(\sigma) = \lambda > 0 & & {\rm with} \ \lambda = \widehat{\lambda}(\sigma)\label{eq:W*1} \\
      &W^*\, ''(\sigma) = \widehat{\lambda}'(\sigma) > 0 & &\forall \sigma \in \mathbb{R} \label{eq:W*2} \\
      &W^*\, '(\sigma) \to 0^+ & & {\rm as} \ \sigma \to -\infty \label{eq:W*3}\\
      &W^*\, '(\sigma) \to +\infty & & {\rm as} \ \sigma \to +\infty \label{eq:W*4}\\
      &W^*(0) = 0 \label{eq:W*5} \\
      &W^*(\sigma) \to \pm \infty & & {\rm as} \ \sigma \to \pm\infty \label{eq:W*6}
    \end{empheq}
  \end{subequations}
  all of which follow from assumptions and definitions previously made in equations \eqref{eq:Wass}-\eqref{eq:W*def}. The key property of $W^*(\sigma)$ is \eqref{eq:W*1}. It ensues from the definition \eqref{eq:W*def} since
  \begin{equation*}
    W^*\, '(\sigma) =  \widehat{\lambda}(\sigma) + \sigma \widehat{\lambda}'(\sigma) - W'(\lambda)\widehat{\lambda}'(\sigma) = \widehat{\lambda}(\sigma) \;.
  \end{equation*}
  Given that $\widehat{\lambda}(\sigma)$ is positive, we have that $W^*(\sigma)$ is a monotonically increasing function. Moreover, \eqref{eq:W*2} follows from \eqref{eq:lhat'} whence $W^*(\sigma)$ is also strictly convex. Properties \eqref{eq:W*3} and \eqref{eq:W*4} are simply restatements of \eqref{eq:lhatprop}. Property \eqref{eq:W*5} follows from the definition of $W^*(\sigma)$ and, together with monotonicity, implies that $W^*(\sigma)>0$ when $\sigma>0$ and $W^*(\sigma)<0$ when $\sigma<0$. The last property \eqref{eq:W*6} follows as well from the definition and the preceding properties. It is important because it implies that $W^*(\sigma): \mathbb{R} \longrightarrow \mathbb{R}$ is surjective, and given the injectivity implied by \eqref{eq:W*1}, also invertible. We will use this result in what follows, so it is worth noticing that it is a consequence of the assumptions \eqref{eq:Wass3} and \eqref{eq:Wass5}.


\section{\label{sec:kinder}Derivation of the driving force}
Accretion is a non-equilibrium process involving dissipation. The latter can be computed as the product of a flux, accretion rates in our case, and of a conjugate driving force which quantifies the departure from thermodynamic equilibrium.

In this section we provide the derivation of the expression of the driving force in eq.\ \eqref{eq:Fa}. We follow \cite{TomassettiCohen:2016} and \cite{AbeyaratneKnowles1990,AbeyaratneKnowles1997}.

We start from the expression of the dissipation rate associated with the bar,
\begin{equation}\label{eq:dissdef}
  \text{dissipation rate} \coloneqq \sigma \Dert{y}{t}\Big|^{x_1}_{x_0} +
                            \varrho(\mu-M_B) \dot{x}\Big|^{x_1}_{x_0} -
                            \Dert{ }{t} \int_{x_0}^{x_1} W(\lambda) d x  \;,
\end{equation}
which involves the sum of three terms. The first term represents the mechanical power of external loads, the second the inflow of chemical energy per unit time and the third the energy flow per unit time elastically stored in the material and therefore not dissipated.

It is crucial to observe that the velocity of the boundary
\begin{equation}\label{eq:vel}
\dot{y}_\alpha(t) =     \Dert{ }{t}y(x_\alpha(t),t) = v_\alpha + y' \dot{x}_\alpha = v_\alpha + \lambda_\alpha \dot{x}_\alpha,
\end{equation}
differs from the velocity \emph{at the boundary}:
\begin{equation}\label{eq:vel2}
  v_\alpha = \Derp{y_\alpha}{t}(x_\alpha(t),t) \,,
\end{equation}
the two quantities being equal only when the growth velocities $\dot{x}_\alpha$ vanish.


We rewrite the third term in \eqref{eq:dissdef} using  Leibnitz's rule (the divergence theorem in one-dimension), transport theorems and equations \eqref{eq:mech1} and \eqref{eq:mech2},
\begin{equation}\label{eq:Wdiv}
  \begin{split}
\Dert{ }{t} \int_{x_0}^{x_1} W(\lambda) d x ={}& \int_{x_0}^{x_1} W'(\lambda) (\dot{y})' d x + W(\lambda) \dot{x}\Big|^{x_1}_{x_0} \\
       ={}& (\sigma \dot{y} + W(\lambda)) \dot{x} \Big|^{x_1}_{x_0} \;.
  \end{split}
\end{equation}

On substituting equations \eqref{eq:vel} and \eqref{eq:Wdiv} into the expression \eqref{eq:dissdef} of the dissipation rate we obtain
\begin{equation}\label{eq:diss2}
\begin{split}
   \text{dissipation rate} ={}&  (\sigma \dot{y} + \sigma \lambda \dot{x} +
                            \varrho(\mu-M_B)\dot{x}
                            - (\sigma \dot{y} + W(\lambda) ) \dot{x} )\Big|^{x_1}_{x_0} \\
     ={}& (\varrho(\mu-M_B) + (\sigma \lambda - W(\lambda))) \; \dot{x} \Big|^{x_1}_{x_0}\\
     ={}& \Big(\varrho(\mu-M_B) + W^*(\sigma)\Big) \dot{x} \Big|^{x_1}_{x_0} \;,
\end{split}
\end{equation}
in which the multiplier of the accretive flux $\dot{x}$ is precisely the driving force of growth introduced in equation \eqref{eq:Fa}.


\section{\label{sec:DAE} Reduction to a differential algebraic system}

Here the system of equations presented in Section \ref{sec:model} is reduced to a differential algebraic system and new notation is introduced, suitable for the ensuing discussion on the existence and stability of treadmilling solutions.

 \subsection{\label{ss:DAEmech} Mechanics}

 Let
 \begin{equation}
  \label{eq:ell}
  \ell(t) = x_1(t) - x_0(t) > 0 \;,
 \end{equation}
 be the length of the bar in the reference configuration.
 The integration of the mechanical system of equations \eqref{eq:mech} yields
 \begin{subequations}\label{eq:msol}
  \begin{empheq}[left=\empheqlbrace]{align}
    &y(x,t) = \lambda(t) (x-x_0(t)) +Y_0, & & \forall \;\; x_0(t) \leq x \leq x_1(t) \label{eq:msol1}\\
    &\sigma(t) = W'(\lambda(t)) \label{eq:msol2}\\
    &K \lambda(t) \ell(t)   = \sigma_{\rm max} - {\sigma(t)} \label{eq:msol3}
  \end{empheq}
 \end{subequations}
 where we have termed
 \begin{equation} \label{eq:smaxdef}
  \sigma_{\rm max} = K(Y_1 - Y_0)
 \end{equation}
 the maximum force attainable in the bar and in the spring. Since both $\lambda >0$ and $\ell >0$, it follows from \eqref{eq:msol3} that
 \begin{equation}\label{eq:s<smax}
  \sigma <  \sigma_{\rm max}.
 \end{equation}
 We consider $\sigma_{\rm max}$ to be an arbitrarily tunable parameter since we can imagine being able to vary the rest position $Y_1$ of the spring, to the right or to the left of $Y_0$, to attain any desired value of $\sigma_{\rm max}$.

 From \eqref{eq:msol1} we have
 \begin{equation}\label{eq:lamell}
   \lambda = (y_1 - y_0)/(x_1 - x_0) = (y_1 - y_0)/\ell \qquad \Rightarrow \qquad (y_1 - y_0) = \lambda \ell
 \end{equation}
 and so, as expected, $\lambda \ell$ denotes the length of the body in physical space.

 Equations \eqref{eq:msol} describe a unique motion $y(x,t)$ and force $\sigma(t)$ in terms of  $x_0$, $x_1$. To see it, combine \eqref{eq:msol2} and \eqref{eq:msol3} to give
 \begin{equation}\label{eq:lamsol}
   W'(\lambda)  = \sigma_{\rm max} - K \ell \lambda \;.
 \end{equation}
 In light of the assumed properties \eqref{eq:Wass} of $W(\lambda)$, it is readily shown that there exists a unique root $\lambda >0$ of this equation corresponding to any given $\ell >0$, $K>0$ and $\sigma_{\rm max}$. Moreover in view of \eqref{eq:Wass6}, the root $\lambda$ decreases monotonically with increasing $\ell$. The corresponding force is then given by \eqref{eq:msol2}. These representations will of course involve given values of $K$, $Y_0$, $Y_1$  and the yet to be found values $x_0$, $x_1$.

 The length $\ell$ of the body in reference space given through \eqref{eq:msol3} can be expressed in terms of force $\sigma$ as
 \begin{equation}\label{eq:ls}
  \ell = \bar{\ell}(\sigma) \coloneqq \frac{\sigma_{\rm max} - \sigma}{K \widehat{\lambda}(\sigma)},
 \end{equation}
  for all $\sigma < \sigma_{\rm max}$, where the function $\widehat{\lambda}(\sigma)$ is the inverse of the force-stretch relation $\sigma = W'(\lambda)$ introduced in eq.\ \eqref{eq:lhat}.
  In view of  \eqref{eq:lhat'}-\eqref{eq:lhatprop}, this shows that
  \begin{equation}\label{eq:lsprop}
    \bar{\ell}{\,}'(\sigma) < 0, \quad
    \bar{\ell}(\sigma) \to +\infty \ {\rm as} \ \sigma \to - \infty, \quad
    \bar{\ell}(\sigma) \to 0^+ \ {\rm as} \ \sigma \to \sigma_{\rm max}^-.
  \end{equation}
  The function $\bar{\sigma}(\ell)$ that is inverse to $\bar{\ell}(\sigma)$ obeys
  \begin{align}\label{eq:sl}
    \bar{\ell}(\bar{\sigma}(\ell)) &= \ell &
    \bar{\sigma}(\ell) &\to \sigma_{\rm max}^-  \quad\text{ as }\; \ell \to 0^+, \notag\\
    \bar{\sigma}'(\ell) &<0 &
    \bar{\sigma}(\ell) &\to - \infty \quad\text{ as }\; \ell \to +\infty \;.
  \end{align}
  %
  From equations \eqref{eq:ls}-\eqref{eq:sl} we can appreciate how growth, i.e. a change of the length of the bar $\ell$ in the reference configuration, affects force and stretch at equilibrium. A decrease in length $\ell$ in \eqref{eq:ls} produces an increase in stretch $\lambda$ and in force $\sigma$ until, eventually, $\ell$ goes to zero, stretch $\lambda$ to infinity and the force to its maximum value $\sigma_{\rm max}$. Conversely, an increase in material length $\ell$ decreases both stretch and force. An indefinite increase of $\ell$ leads the force towards infinite compressive values and stretch towards zero.


  \subsection{\label{ss:DAEdiff} Diffusion}

  The solution of the system of equations \eqref{eq:diff} yields
  \begin{subequations}\label{eq:dsol}
   \begin{empheq}[left=\empheqlbrace]{align}
     &\mu(y,t) = M_1\frac{y-y_0}{y_1 - y_0} + \mu_0 \frac{y_1 - y}{y_1 - y_0}, & & \forall \;\; y_0(t) \leq y \leq y_1(t) \label{eq:dsol1}\\
     &h(y,t) = - m \frac{M_1 - \mu_0}{y_1 - y_0} \label{eq:dsol2}\\
     &\mu_0 = M_1 + \frac{\varrho}{m}\,(y_1 - y_0) \,\dot{x}_0 \label{eq:dsol3}
   \end{empheq}
  \end{subequations}
  and we recall that $y_1 - y_0 = \lambda \ell$. Using \eqref{eq:msol3} and \eqref{eq:lamell} one can express $\mu_0$ in terms of the force $\sigma$,
  \begin{equation}\label{eq:mu0s}
    \mu_0 = M_1 + \frac{\varrho}{K m} (\sigma_{\rm max}-\sigma) \,\dot{x}_0 \;.
  \end{equation}
  Observe that \eqref{eq:mu0s} can be used to eliminate the unknown chemical potential $\mu_0$ from the other equations where its appears, namely \eqref{eq:accr1}, \eqref{eq:dsol1} and \eqref{eq:dsol2}.

  Finally we note that if $x_0$ and $x_1$ are known, then as noted previously $y_1$ can be determined from \eqref{eq:ell}, \eqref{eq:lamell} and \eqref{eq:lamsol}, $y_0=Y_0$ being of course known. If in addition $\dot{x}_0$ is known then the chemical potential and flux fields are fully determined through \eqref{eq:dsol}.

  \subsection{\label{ss:DAEaccr} Accretion}

  Using \eqref{eq:mu0s}, setting $\mu_1=M_1$ and noting that $\sigma_0(t) = \sigma_1(t) = \sigma(t)$, we rewrite the pair of kinetic equations \eqref{eq:accr} as
  \begin{subequations}\label{eq:asub}
    \begin{empheq}[left=\empheqlbrace]{align}
      \dot x_0(t)&=-\frac{1}{B_0}\frac{\varrho(M_1-M_{B,0})+W^*(\sigma(t))}
      {1+\frac{\varrho^2}{m B_0 K} (\sigma_{\rm max}-\sigma)} , \label{eq:asub1}\\
      \dot x_1(t)&=\frac{1}{B_1}\brac{\varrho(M_1-M_{B,1})+W^*(\sigma(t))} \,.\label{eq:asub2}
    \end{empheq}
  \end{subequations}
  We now introduce forces $\sigma_{\alpha 0}$, $\sigma_{\alpha 1}$ exploiting the bijectivity of $W^*(\sigma)$ in $\mathbb{R}$ at which the accretion rates $\dot x_0(t)$, $\dot x_1(t)$ vanish:
  \begin{equation}\label{eq:salpha}
    \sigma_{\alpha 0} \, : \; -W^*(\sigma_{\alpha 0}) = \varrho(M_1-M_{B,0}) \;, \quad
    \sigma_{\alpha 1} \, : \; -W^*(\sigma_{\alpha 1}) = \varrho(M_1-M_{B,1}) \,.
  \end{equation}
Because of the monotonicity of $W^*$ we see that $\sigma_{\alpha 0} < \sigma_{\alpha 1}$ if $M_{B,0} < M_{B,1}$ and vice versa. As to which of these holds will play a central role in Section \ref{ss:TMexist} when we look at the existence of treadmilling states.  Finally let the forces $\Delta\sigma$, $\sigma_{\rm asymp}$ be defined by
  \begin{equation}\label{eq:sasym}
    \Delta\sigma \coloneqq \sigma_{\rm asymp}-\sigma_{\rm max} \coloneqq \frac{m B_0 K}{\varrho^2} > 0 \;.
  \end{equation}
  This allows us to write
  \begin{subequations}\label{eq:afin}
    \begin{empheq}[left=\empheqlbrace]{align}
      \dot x_0(t)&= R_0(\sigma)  \coloneqq -\frac{\Delta\sigma}{B_0}\;\frac{W^*(\sigma(t))-W^*(\sigma_{\alpha 0})}
      {\sigma_{\rm asymp}-\sigma} , \label{eq:afin1}\\
      \dot x_1(t)&= R_1(\sigma) \coloneqq \frac{1}{B_1}\,\brac{W^*(\sigma(t))-W^*(\sigma_{\alpha 1})} \,, \label{eq:afin2}
    \end{empheq}
  \end{subequations}
  for the accretion rates $\dot x_0(t)$, $\dot x_1(t)$ as functions $R_0(\sigma)$ and $R_1(\sigma)$ of the force, respectively.

Since the chemical potential of the solvent bath $M_1$ can be varied, according to their definitions \eqref{eq:salpha}, the values of $\sigma_{\alpha 0}$ and $\sigma_{\alpha 1}$ may also be varied, but not independently. In addition, for the admissible values \eqref{eq:s<smax} of the force $\sigma$ smaller than $\sigma_{\rm max}$, relation \eqref{eq:sasym} and the monotonicity \eqref{eq:W*1} of $W^*(\sigma)$ tell us that
  \begin{equation}\label{eq:R>0}
    R_0(\sigma) \lesseqqgtr 0 \;\text{ for } \sigma \gtreqqless \sigma_{\alpha 0} \text{ and } \sigma < \sigma_{\rm max}\,, \quad
    R_1(\sigma) \gtreqqless 0 \;\text{ for } \sigma \gtreqqless \sigma_{\alpha 1} \,,
  \end{equation}
  underscoring in particular that $\sigma_{\alpha 0}$ and $\sigma_{\alpha 1}$ are the unique zeros of $R_0(\sigma)$ and $R_1(\sigma)$ respectively.

  For $R_1(\sigma)$ we can easily infer its properties from those of $W^*$: it is a convex, monotonically increasing function whose image is all $\mathbb{R}$ and whose derivative tends to $0^+$ for $\sigma \to -\infty$ and to $+\infty$ for $\sigma \to +\infty$.

  Of $R_0(\sigma)$ we know that it tends to $-\infty$ as $\sigma$ approaches $\sigma_{\rm asymp}$ from below in the case where $\sigma_{\rm asymp} > \sigma_{\alpha 0}$.  Using l'Hopital's rule we also deduce that $R_0(\sigma)$ tends to zero as $\sigma \to -\infty$.
  Looking at the first derivative of $R_0(\sigma)$,
  \begin{equation}\label{eq:R0'}
    R_0'(\sigma)  = -\frac{\Delta\sigma}{B_0}\;\frac{W^*(\sigma)-W^*(\sigma_{\alpha 0})
    +W^*\,'(\sigma) (\sigma_{\rm asymp}-\sigma)}
      {(\sigma_{\rm asymp}-\sigma)^2} \;,
  \end{equation}
  we observe that it is strictly negative for
  $\sigma_{\alpha 0} < \sigma < \sigma_{\rm asymp}$. Instead, $R_0$ is not monotonic for $\sigma < \sigma_{\alpha 0} < \sigma_{\rm asymp}$, since its derivative has opposite signs at the two ends $\sigma_{\alpha 0}$, $-\infty$ of that interval.


  \subsection{\label{ss:DAEeq} Differential algebraic equation}

 By \eqref{eq:W*1}, \eqref{eq:ls} and  \eqref{eq:afin} the model under consideration governing the evolution of the length $\ell$ of the bar in the reference configuration reduces to the following differential algebraic equation
  \begin{subequations}\label{eq:lsdae}
    \begin{empheq}[left=\empheqlbrace]{align}
      \dot \ell &= R_1(\sigma) - R_0(\sigma) \notag \\
      &=
      \frac{1}{B_1}\,\brac{W^*(\sigma)-W^*(\sigma_{\alpha 1})}
      +\frac{\Delta\sigma}{B_0}\;\frac{W^*(\sigma)-W^*(\sigma_{\alpha 0})}
      {\sigma_{\rm asymp}-\sigma} ,\label{eq:lsdae1}\\
      \ell &= \bar{\ell}(\sigma) = \frac{\sigma_{\rm max} - \sigma}{K \,\, W^* \,'(\sigma)}  \,, \label{eq:lsdae2}
    \end{empheq}
  \end{subequations}
  in which $\ell(t)$ and $\sigma(t)$ are sought under initial conditions, see \eqref{eq:xyinit} and \eqref{eq:lamell},
  \begin{equation}\label{eq:lsinit}
    \ell(0) = x_{10}-x_{00}, \quad \sigma(0)= W'\brac{\frac{Y_{10}-Y_0}{x_{10}-x_{00}}} \;,
  \end{equation}
  and under the constraint $\ell > 0$ which is equivalent to $\sigma < \sigma_{\rm max}$.


\section{\label{sec:TM} Existence and stability of treadmilling solutions}

In a so-called treadmilling solution the length $\ell = \ell_{\rm TM}$ of the bar in the reference configuration does not vary with time: $\dot{\ell}=0$, and this corresponds to values $\sigma_{\rm TM}$ of the force for which $R_0(\sigma_{\rm TM}) = R_1(\sigma_{\rm TM})$.

In the rest of this section we investigate the existence and stability of treadmilling solutions.


\subsection{\label{ss:TMstab} Stability of treadmilling states}


We start by discussing the stability of a treadmilling solution, assuming one to exist, by perturbing a treadmilling state characterized by force $\sigma_{\rm TM}$ and length $\ell_{\rm TM} = \bar{\ell}(\sigma_{\rm TM})$ according to \eqref{eq:lsdae2}.

The perturbation of equation \eqref{eq:lhat}, $\lambda = \widehat{\lambda}(\sigma)$, yields
\begin{equation*}
  \delta \lambda = \widehat{\lambda}'(\sigma_{\rm TM}) \delta \sigma .
\end{equation*}
Operating analogously on equation \eqref{eq:msol3}, $K \lambda \ell = \sigma_{\rm max} - \sigma$, gives
\begin{equation*}
  K \widehat{\lambda}(\sigma_{\rm TM}) \delta \ell + K \ell_{\rm TM} \delta \lambda = - \delta \sigma .
\end{equation*}
Combining the two preceding equations provides a relation between the perturbations $\delta \ell$ and $\delta \sigma$,
\begin{equation}\label{eq:psell}
K \widehat{\lambda}(\sigma_{\rm TM}) \, \delta \ell + \brac{K \ell_{\rm TM} \widehat{\lambda}'(\sigma_{\rm TM}) + 1} \, \delta \sigma = 0.
\end{equation}
From the expression \eqref{eq:lsdae1} of $\dot{\ell}(\sigma)$ we obtain
\begin{equation*}
  \delta \dot{\ell} = \brac{R_1'(\sigma_{\rm TM}) - R_0'(\sigma_{\rm TM})} \, \delta \sigma.
\end{equation*}
Combining this with \eqref{eq:psell} yields
\begin{equation*}
  \delta \dot{\ell} = - F(\sigma_{\rm TM}) \delta \ell \qquad {\rm where} \quad F(\sigma_{\rm TM}) =
  \frac{ R_1'(\sigma_{\rm TM}) - R_0'(\sigma_{\rm TM}) }{ K \ell_{\rm TM} \widehat{\lambda}'(\sigma_{\rm TM}) + 1 } \, K \widehat{\lambda}(\sigma_{\rm TM}).
\end{equation*}
The treadmilling solution is stable if the linear ordinary differential equation $\delta \dot{\ell}(t) = - F(\sigma_{\rm TM}) \delta \ell(t)$ has exponentially decaying solutions\footnote{If $\delta \ell(t)$ vanishes exponentially then so do the perturbations of the various other quantities.}, and this  occurs if and only if $F(\sigma_{\rm TM}) >0$. Since $\widehat{\lambda}(\sigma_{\rm TM}) >0$, $\ell_{\rm TM} >0$ and $\widehat{\lambda}' (\sigma_{\rm TM}) >0$ it follows that the
treadmilling solution is stable if and only if
\begin{equation} \label{eq:TMstable}
R_1'(\sigma_{\rm TM}) > R_0'(\sigma_{\rm TM}) .
\end{equation}
It is interesting to use the monotonic relation between the force $\sigma$ and the referential length $\ell$ to plot the evolution of the system in the neighborhood of a stable treadmilling solution on the $\ell, \dot{\ell}$-plane. The system \eqref{eq:lsdae} can be written using
\eqref{eq:sl} as $\dot{\ell} = R_1(\bar\sigma(\ell)) -  R_0(\bar\sigma(\ell))$.  Given that $\bar{\sigma}(\ell)$ is a monotonically decreasing function,  the slope of $\dot{\ell}(\ell)$ is thus seen to be negative close to a stable treadmilling in view of \eqref{eq:TMstable}.

This is shown schematically in Figure \ref{fig:ell}. Observe how, if $\ell > \ell_{\rm TM}$ at some time then $\dot{\ell} <0$ and so $\ell(t)$ will decrease until it reaches the treadmilling value  $\ell_{\rm TM}$. Likewise if $\ell < \ell_{\rm TM}$, $\ell(t)$ will increase to $\ell_{\rm TM}$.
\begin{figure}[h]
\begin{center}
\includegraphics[width=0.5\linewidth]{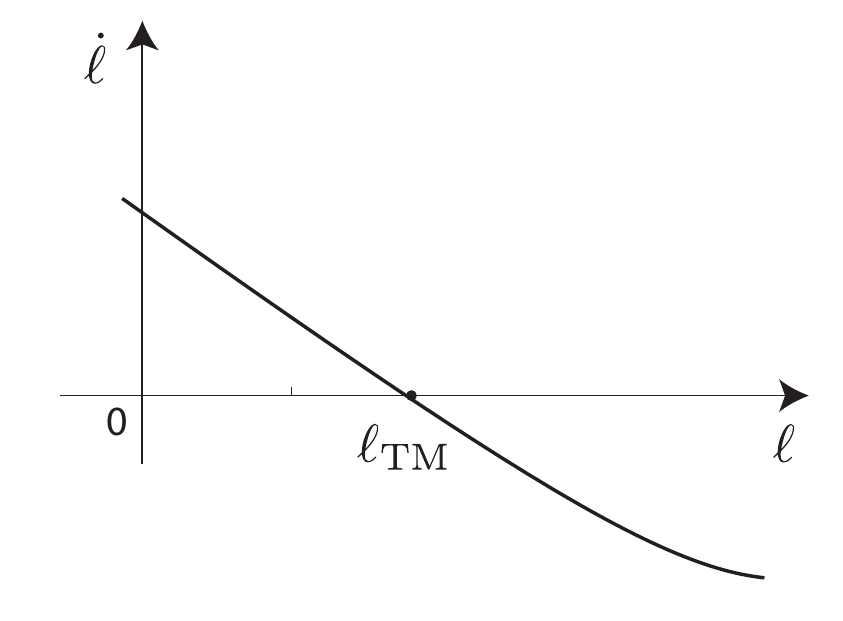}
\end{center}
\caption{\footnotesize Graph of  $\dot{\ell} = \dot{x}_1 - \dot{x}_0 = R_1(\bar{\sigma}(\ell)) - R_0(\bar{\sigma}(\ell))$ versus $\ell$ where $\bar{\sigma}(\ell)$ is given through \eqref{eq:sl}.}
\label{fig:ell}
\normalsize
\end{figure}


\subsection{\label{ss:TMexist} Existence of treadmilling solutions}

We are going to show that a relevant condition for the existence of a treadmilling solution is that accretion prevails over ablation in the limit $\ell \to 0$. In view of \eqref{eq:ls} this is equivalent to $\sigma \to \sigma_{\rm max}$ and hence, by \eqref{eq:lsdae1} to
\begin{equation}\label{eq:R1>R0}
  R_1(\sigma_{\rm max}) > R_0(\sigma_{\rm max}) \, .
\end{equation}
%
It is possible to formulate \eqref{eq:R1>R0} in a more convenient form by introducing the force $\sigma_\beta$ according to
\begin{equation} \label{eq:sb}
   W^*(\sigma_\beta) =
    \frac{1}{1+\beta} W^*(\sigma_{\alpha 0}) + \frac{\beta}{1+\beta} W^*(\sigma_{\alpha 1}) \,, \quad \text{with} \quad \beta = \frac{B_0}{B_1} > 0 \; .
\end{equation}
From the monotonicity of $W^*$, it is clear that the value of $\sigma_\beta$ always lies between $\sigma_{\alpha 0}$ and $\sigma_{\alpha 1}$ defined in \eqref{eq:salpha}.
 It is readily seen that the inequality \eqref{eq:R1>R0} is equivalent to
\begin{equation}\label{eq:sb<smax}
  \sigma_\beta < \sigma_{\rm max} \; .
\end{equation}
The inequality \eqref{eq:sb<smax} will play a central role in the results to follow.

\subsubsection{\label{sss:TM0<1} Case $M_{B,0}$ smaller than $M_{B,1}$:}

We discuss first the case that is more relevant with regard to applications \citep{TomassettiCohen:2016,Abi-AklAbeyaratne2019} and experiments \citep{CameronFooter1999,NoireauxGolsteyn:2000,GuchtPaluch2005,BielingLi2016} where the treadmilling state involves accretion at the fixed end $Y_0$ and ablation at the free end $y_1$ (rather than the converse).

Suppose that
\begin{equation} \label{eq:MB0<MB1}
  M_{B,0} < M_{B,1} \,,
\end{equation}
corresponding to ATP-actin being at $Y_0$ and the hydrolyzed ADP-actin at $y_1$. From the definition \eqref{eq:salpha} and the monotonicity \eqref{eq:W*1} of $W^*$,  it follows that  \eqref{eq:MB0<MB1} holds if and only if
\begin{equation}\label{eq:sa0<sa1}
  \sigma_{\alpha 0} < \sigma_{\alpha 1} \,.
\end{equation}

Under the above provision, it is possible to prove that,
\begin{prop}[Existence and uniqueness] \label{prop-1} Assume that \eqref{eq:MB0<MB1} (equivalently \eqref{eq:sa0<sa1}) holds. Then
  $\sigma_\beta < \sigma_{\rm max}$ is a necessary and sufficient condition for the existence and uniqueness of a treadmilling solution. Such a solution is globally stable.
\end{prop}
\begin{proof}
To show sufficiency,  assume $\sigma_\beta < \sigma_{\rm max}$ holds. Since $\sigma_\beta$ lies between $ \sigma_{\alpha 0}$ and $\sigma_{\alpha 1}$ it now follows that in the present case
\begin{equation}
\label{eq:stressineq}
  \sigma_{\alpha 0} < \sigma_\beta < \sigma_{\rm max} < \sigma_{\rm asymp}.
\end{equation}
Moreover, as observed in section \ref{ss:DAEaccr}, in the interval $\sigma_{\alpha 0} < \sigma < \sigma_{\rm max}$
  \begin{itemize}
    \item $R_0(\sigma)$ is negative,
    \item $R_0(\sigma)$ is monotonically decreasing,
    \item $R_1(\sigma)$ is monotonically increasing,
    \item $R_1(\sigma_{\alpha 0}) < R_0(\sigma_{\alpha 0}) = 0$,
    \item $R_1(\sigma_{\rm \max}) > R_0(\sigma_{\rm \max}) $.
  \end{itemize}
 Given the continuity of $R_0(\sigma)$ and $R_1(\sigma)$, and because their difference has opposite signs at the extremes of the interval, it follows that they have the same value at some $\sigma = \sigma_{\rm TM}$ in this interval. Uniqueness follows from monotonicity. The negative sign of $R_0(\sigma_{\rm TM})$ implies that $R_1$ too is negative at $\sigma_{\rm TM}$, and therefore the treadmilling force lies in the interval $\sigma_{\alpha 0} < \sigma_{\rm TM} < \min(\sigma_{\alpha 1},\sigma_{\rm max})$. Moreover, since $\dot x_0 = \dot x_1 = R_0(\sigma_{\rm TM}) = R_1(\sigma_{\rm TM}) <0$ at treadmilling, accretion takes place at $Y_0$ and ablation at $y_1$. From monotonicity we also additionally infer that condition \eqref{eq:TMstable} is met and that the unique treadmilling solution is always stable. In fact it is globally stable because the rate $\dot{\ell} = R_1(\sigma) - R_0(\sigma)$ is always negative for all $\sigma < \sigma_{\rm TM}$, i.e. for $\ell > \ell_{\rm TM}$, and vice versa $\dot{\ell} > 0$ for $\sigma_{\rm TM} < \sigma < \sigma_{\rm max}$.

 To show necessity,  assume conversely that $\sigma_\beta \geq \sigma_{\rm max}$. The case $\sigma_\beta = \sigma_{\rm max}$, i.e. $R_1(\sigma_{\rm max}) = R_0(\sigma_{\rm max})$, corresponds to a non admissible treadmilling solution in which the force $\sigma_{\rm TM}$ equals $\sigma_{\rm max}$ and, according to \eqref{eq:ls}, the bar has zero reference length $\ell$.  

 Consider the case $\sigma_\beta > \sigma_{\rm max}$. According to \eqref{eq:R>0}, there is no intersection in the interval $\sigma < \min(\sigma_{\alpha 0},\sigma_{\rm max})  < \sigma_{\alpha 1}$, because $R_0(\sigma)$ is positive and $R_1(\sigma)$ negative. This ends the proof if $\sigma_{\alpha 0} \geq \sigma_{\rm max}$. If $\sigma_{\alpha 0} < \sigma_{\rm max}$, we consider the interval $\sigma_{\alpha 0} \leq \sigma < \sigma_{\rm max}$: both $R_0(\sigma)$ and $R_1(\sigma)$ are monotonic and $R_0$ is above $R_1$ at both ends of the interval. Hence any intersection is excluded also in this case.

\end{proof}

An example of functions $R_0$ and $R_1$ in the case $\sigma_{\alpha 0} < \sigma_{\alpha 1}$ is shown in Figure~\ref{fig:tm1}.
\begin{figure}[h]
\begin{center}
\includegraphics[width=0.75\linewidth]{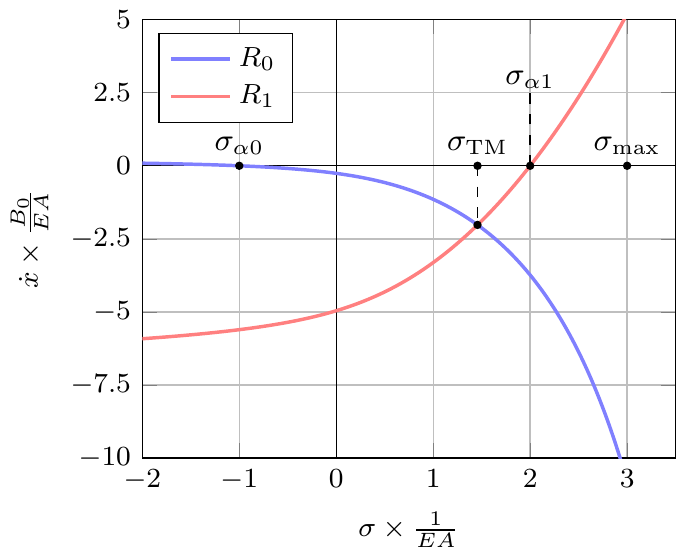}
\end{center}
\caption{\footnotesize Example of treadmilling solution obtained using $W^*(\sigma)$ from eq.~\eqref{eq:W*ex2} with parameters  $\sigma_{\rm asymp} = 5 EA$, $\sigma_{\rm max} = 3 EA$, $\sigma_{\alpha 0} = - EA$, $\sigma_{\alpha 1} = 2 EA$, $\beta = 1.0$.}
\label{fig:tm1}
\normalsize
\end{figure}
There we see that at $\sigma = \sigma_{\rm max}$, condition \eqref{eq:R1>R0}, is met and therefore a unique, globally stable, treadmilling solution exists. This result corroborates the globally stable equilibrium observed numerically by \cite{Abi-AklAbeyaratne2019} in a similar setting.

Once the value of the force $\sigma_{\rm TM}$ at treadmilling is known, it is possible to reconstruct the whole system at treadmilling. The corresponding stretch $\lambda_{\rm TM}$ is given by $\lambda_{\rm TM} = \widehat{\lambda}(\sigma_{\rm TM})$; the growth rates at the two ends are $\dot{x}_0^{\rm TM} = \dot{x}_1^{\rm TM} = R_0(\sigma_{\rm TM}) = R_1(\sigma_{\rm TM})$; the length of the body in reference space is $\ell_{\rm TM} = \bar{\ell}(\sigma_{\rm TM})$; the length of the body in physical space is $\lambda_{\rm TM} \ell_{\rm TM}$; and the chemical potential at the growing end is $\mu_0^{\rm TM} = M_1 + \varrho (\sigma_{\rm max} - \sigma_{\rm TM}) \dot{x}_0^{\rm TM}/(K m)$.

If \eqref{eq:sa0<sa1} holds but \eqref{eq:sb<smax} does not, then differently from the situation represented in Figure \ref{fig:tm1}, the value of $\sigma_{\rm max}$ is smaller than $\sigma_{\rm TM}$, the point where $R_0$ and $R_1$ intersect, and so  there is no treadmilling solution. According to \eqref{eq:lsdae1}, we have $\dot \ell$  negative everywhere in the admissible domain, so  if the bar has some nonzero length at the initial instant, then it progressively loses all of its monomers till it reaches $\ell = 0$ and $\sigma= \sigma_{\rm max}$.

It is interesting to represent the evolution in space (Figure~\ref{fig:tm-eny}) and time (Figure~\ref{fig:tment}) of the energy of a material unit of actin as it undergoes treadmilling.

The energy $e$ of a material unit is comprised of its chemical potential $\mu$, elastic strain energy $W$ and the potential energy of the stress $\sigma \lambda$: $e = \rho \mu + W - \sigma \lambda = \rho \mu - W^*$. In Figures~\ref{fig:tm-eny} and  ~\ref{fig:tment} the points $(f)$ and $(a)$ correspond to states just before and after accretion, while $(d)$ and $(e)$ refer to just before and after ablation. The energy of a free monomer at $Y_0$, just before accretion, is $e_{f} = \rho \mu_0$, and just after accretion, when it is bound to the solid, its energy is $e_a = \rho M_{B,0} - W^*(\sigma_0)$. The difference between these two energies is precisely the driving force at $Y_0$: $F_0 = e_f - e_a$. The dissipation inequality requires $F_0 V_0 = - F_0 \dot x_0 \geq 0$. When accretion takes place $\dot x_0 <0$ and therefore $F_0 \geq 0$ or equivalently $e_a \leq e_f$.  Likewise the energy of a material unit at $y_1$, when it is still bound to the solid, is $e_d = \rho M_{B,1} - W^*(\sigma_1)$ and when it is free after ablation it is $e_e =\rho \mu_1$. The corresponding driving force is $F_1 =e_e- e_d$ and the dissipation inequality requires $F_1V_1 = F_1 \dot x_1 \geq 0$. When ablation takes place at $y_1$, $\dot x_1 < 0$ and so $F_1 \leq 0$ and $e_e \leq e_d$.

\begin{figure}[h]
\footnotesize
\begin{center}
\includegraphics[width=0.75\linewidth]{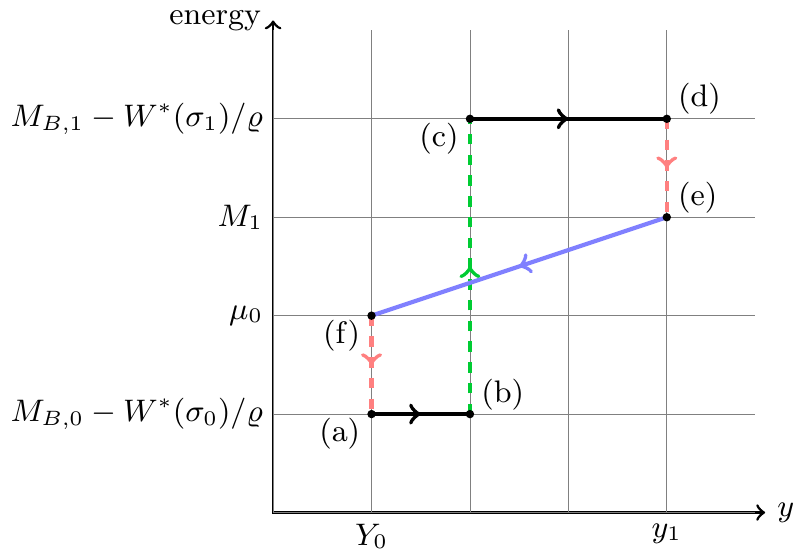}
\caption{Evolution in space of the energy of a mole of actin as it undergoes treadmilling. Blue represents diffusion, red accretion/ablation and green ATP-hydrolysis.} \label{fig:tm-eny}
\end{center}
\normalsize
\end{figure}
A unit of ATP-actin when it is a free monomer in the monomer pool and is located at the right-hand end of the bar corresponds to point $(e)$ in Figure~\ref{fig:tm-eny}.  As it diffuses through the bar it goes from $(e) \to (f)$ following the blue path, and eventually  arrives at the left-hand end of the bar. Accretion then takes place as described in the preceding paragraph and the material unit follows the path $(f) \to (a)$ following the red dashed line, losing energy in the process. The actin unit that got attached to the solid at $(a)$ is now progressively pushed outwards along the bar (due to the accretion of additional material at the left end) and moves towards the right-hand end. This corresponds to $(a) \to (b) \to (c) \to (d)$ with the green segment $(b) \to (c)$ being associated with the hydrolysis step where ATP is converted to ADP which has a higher chemical potential. The location along the bar at which hydrolysis takes place is not determined in our model.  Once this ADP-actin unit reaches the right-hand end  at $(d)$ it undergoes ablation and is detached from the solid following the red dashed line accompanied by an energy loss. It has now returned to the starting point $(e)$  and the process starts again.
\begin{figure}[h]
\footnotesize
\begin{center}
\includegraphics[width=0.75\linewidth]{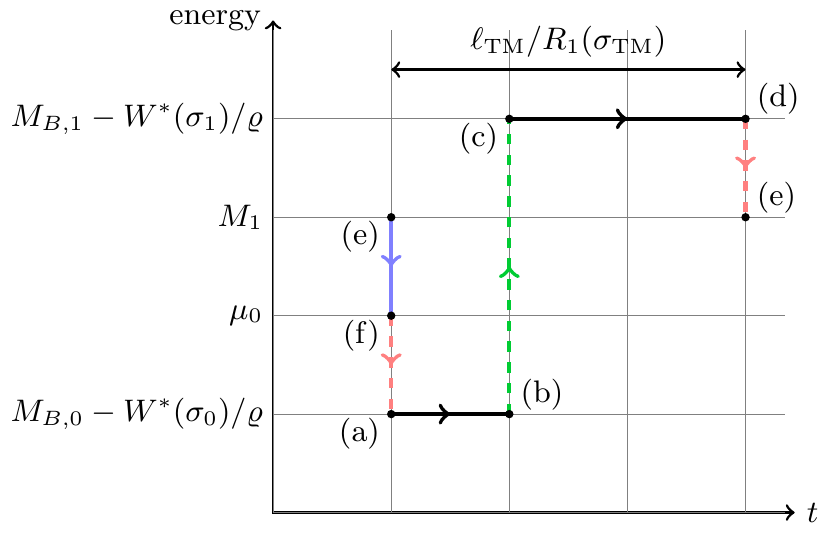}
\caption{Evolution in time of the energy of a mole of actin as it undergoes treadmilling. Blue represents diffusion, red accretion/ablation and green ATP-hydrolysis.}\label{fig:tment}
\end{center}
\normalsize
\end{figure}
The same evolution is represented in Figure~\ref{fig:tment}, with the difference that, consistently with the model assumptions, the instantaneous diffusion segment is vertical.


\subsubsection{\label{sss:TM0>1}  Case $M_{B,0}$ larger than $M_{B,1}$.}

We now discuss the case in which at treadmilling accretion takes place at the free end $y_1$, and ablation at the fixed end $Y_0$. This situation could resemble the experimental set-ups of \cite{ParekhChaudhuri2005} and \cite{ChaudhuriParekh2007}, though treadmilling was not investigated there.

Suppose that
\begin{equation} \label{eq:MB0>MB1}
  M_{B,0} > M_{B,1},
\end{equation}
that is, ATP-actin is at $y_1$ and the hydrolyzed ADP-actin at $Y_0$. From the definition \eqref{eq:salpha} and from the monotonicity \eqref{eq:W*1} of $W^*$, it follows that \eqref{eq:MB0>MB1} holds if and only if
\begin{equation}\label{eq:sa0>sa1}
  \sigma_{\alpha 0} > \sigma_{\alpha 1} \,.
\end{equation}

Under the above condition, it is possible to prove that,
\begin{prop}[Existence] \label{prop-2} Given that \eqref{eq:MB0>MB1} (equivalently \eqref{eq:sa0>sa1}) holds,
  $\sigma_\beta < \sigma_{\rm max}$ is a sufficient condition for the existence of a treadmilling solution.
\end{prop}
\begin{proof}
  Assume that $\sigma_\beta < \sigma_{\rm max}$ holds.
  Then since $\sigma_{\alpha 1} < \sigma_\beta < \sigma_{\alpha 0}$ it follows that $\sigma_{\alpha 1} < \sigma_{\rm max}$.
  Moreover, as observed in section \ref{ss:DAEaccr}, in the interval $\sigma_{\alpha 1} < \sigma < \sigma_{\rm max}$
  \begin{itemize}
    \item $R_1(\sigma)$ is positive,
    \item $R_0(\sigma_{\alpha 1}) > 0 = R_1(\sigma_{\alpha 1})$  since $\sigma_{\alpha 1} < \sigma_{\alpha 0}$, and
    \item $R_1(\sigma_{\rm max}) > R_0(\sigma_{\rm max})$ by \eqref{eq:R1>R0} and \eqref{eq:sb<smax}.
  \end{itemize}
  Given the continuity of $R_0(\sigma)$ and $R_1(\sigma)$, the two functions certainly intersect at least once between $\sigma_{\alpha 1}$ and $\sigma_{\rm max}$ because their difference has opposite signs at the extremes of the interval.  Positiveness of $R_1(\sigma)$, implies that both $R_0(\sigma_{\rm TM})$ and $R_1(\sigma_{\rm TM})$ are positive, and so $\dot x_0 = \dot x_1 >0$ corresponding to ablation at $Y_0$ and accretion at $y_1$. Such a treadmilling state may only exist in the interval $\sigma_{\alpha 1} < \sigma_{\rm TM} < \min(\sigma_{\alpha 0},\sigma_{\rm max})$.
  The solution need not be unique because $R_0$ is not necessarily monotonic for $\sigma < \sigma_{\alpha 0}$.

\end{proof}



An example of functions $R_0$ and $R_1$ in the case $\sigma_{\alpha 0} > \sigma_{\alpha 1}$ is shown in Figure~\ref{fig:tm2}.
\begin{figure}[ht]
\begin{center}
\includegraphics[width=0.75\linewidth]{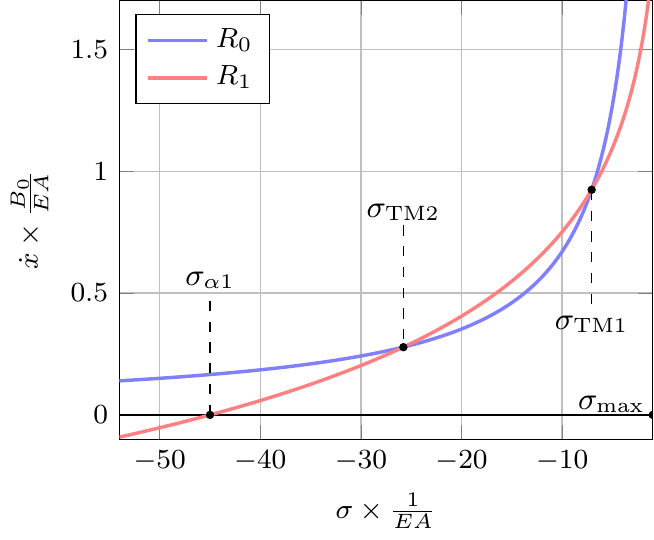}
\end{center}
\caption{\footnotesize Example of treadmilling solution obtained using $W^*(\sigma)$ from eq.~\eqref{eq:W*ex2} with parameters  $\sigma_{\rm asymp} = 0$, $\sigma_{\rm max} = - EA$, $\sigma_{\alpha 0} = 2 EA$, $\sigma_{\alpha 1} = -45 EA$, $\beta = 1.0$. }
\label{fig:tm2}
\normalsize
\end{figure}
The figure has been drawn for a case in which $\sigma_{\rm asymp} < \sigma_{\alpha 0}$ and in which, correspondingly, both $R_0$ and $R_1$ are monotonically increasing in the domain $\sigma < \sigma_{\rm max}$. Notice that multiple treadmilling solutions exist despite conditions \eqref{eq:R1>R0}, \eqref{eq:sb<smax} not being met, because in the current case, $\sigma_{\alpha 1} < \sigma_{\alpha 0}$, such conditions are sufficient but not necessary for the existence of treadmilling solutions. According to the local stability criterion \eqref{eq:TMstable}, the treadmilling state characterized by force $\sigma_{{\rm TM} 1}$ is unstable while the other at force $\sigma_{{\rm TM} 2}$ is stable. The basin of attraction of the latter solution is any initial condition corresponding to a force $\sigma < \sigma_{{\rm TM} 1}$. A bar whose initial length $\ell(0) < \ell_{{\rm TM} 1}$, i.e. $\sigma(0) > \sigma_{{\rm TM} 1}$, will have negative $\dot{\ell}(t)$ and therefore progressively lose all of its monomers  and approach zero length and force $\sigma_{\rm max}$.


\section{\label{sec:disc} Discussion of the results}

\subsection{\label{ss:ellM} Variation of length with chemical potential}

It is illuminating to interpret the main results of this study, Propositions \ref{prop-1} and \ref{prop-2}, in terms of the chemical potential $M_1$ of free monomers and the referential length $\ell$ of the bar.  In a laboratory experiment one can imagine varying $M_1$ and observing how $\ell$ changes.

\begin{figure}[h]
\centerline{\includegraphics[width=0.75\linewidth]{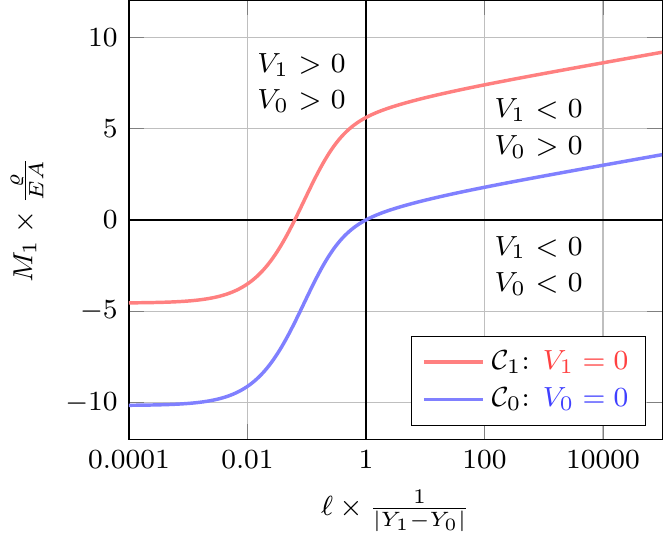}}

\caption{
The $\ell, M_1$-plane in the case $M_{B,1}> M_{B,0}$. The growth rates $V_0$ and $V_1$ vanish on the respective curves ${\cal C}_0$ and ${\cal C}_1$. These curves demarcate 3 regions of the $\ell, M_1$-plane where the signs of the growth velocities  are as shown. $M_{B,1} = 0$. Other parameters as in Figure~\ref{fig:tm1}. }
\label{fig:LM1a}
\end{figure}


%
%

First consider how each end of the bar grows in different regions of the $\ell, M_1$-plane. The curves ${\cal C}_0$ and ${\cal C}_1$ at which $V_0=0$ and $V_1=0$ are shown in Figure \ref{fig:LM1a} in the case
$M_{B,1}> M_{B,0}$ (the case of principal interest). By \eqref{eq:asub} and \eqref{eq:sl} they are characterized by
\begin{equation}
\label{eq:0502-eq4}
{\cal C}_0: \ M_1 = M_{B,0} - W^*(\overline{\sigma}(\ell))/\varrho, \qquad
{\cal C}_1: \ M_1 = M_{B,1} - W^*(\overline{\sigma}(\ell))/\varrho.
\end{equation}
That these curves have the monotonicity depicted in the figure follows from the properties of $\overline{\sigma}(\ell)$ and $W^*(\sigma)$ which tell us that $- W^*(\overline{\sigma}(\ell))$ increases monotonically from $W^*(\sigma_{\rm max})$ to $+ \infty$ as $\ell$ increases from $\ell =0$.  These curves demarcate 3 regions of the $\ell, M_1$-plane where the signs of the growth velocities $V_0 = - \dot x_0$ and $V_1 = \dot x_1$ are as shown. Observe that if the bar is sufficiently long or the chemical potential of the free monomers is sufficiently small, corresponding to points on the right of ${\cal C}_0$, ablation happens at both ends of the bar ($V_1<0, V_0 < 0$) and it will grow shorter. On the other hand on the left of ${\cal C}_1$, where the bar is sufficiently short or the chemical potential is sufficiently large, accretion happens at both ends ($V_1>0, V_0 > 0$) and the bar will grow longer. Between the two curves ${\cal C}_0$ and ${\cal C}_1$, where the length $\ell$ of the bar and the chemical potential $M_1$ have intermediate values, accretion occurs at the left-hand end ($V_0 > 0$) and ablation occurs at the right-hand end ($V_1<0$).

\begin{figure}[h]
\centerline{\includegraphics[width=1.0\linewidth]{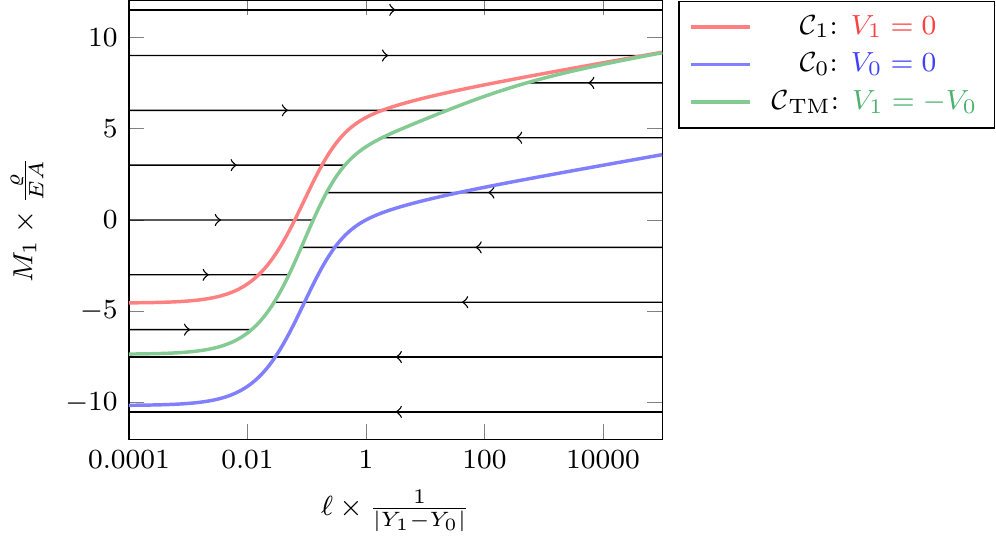}}

\caption{
A point $(\ell, M_1)$ on the curve ${\cal C}_{\rm TM}$ corresponds to a treadmilling state. The arrow at an arbitrary point $(\ell, M_1)$ tells us whether $\dot\ell$ is positive or negative there thus indicating whether the length of the bar increases or decreases. Observe from the figure that a treadmilling solution exists whenever \eqref{eq:0502-eq1} holds. $M_{B,1} = 0$. Other parameters as in Figure~\ref{fig:tm1}.}

\label{fig:LM1b}
\end{figure}


The curve ${\cal C}_{\rm TM}$ corresponding to treadmilling is found by setting $V_0 = - V_1 \, (\neq 0)$, i.e. $\dot{x}_0 = \dot{x}_1$, which from \eqref{eq:asub} is found to be described by
\begin{equation}
\label{eq:0502-eq2}
{\cal C}_{\rm TM}: \  M_1 = \overline{M}_1(\ell) \coloneqq  \frac{M_{B,0} + \beta\left[ 1 + \frac{\rho^2}{mB_0K}(\sigma_{\rm max} - \overline{\sigma}(\ell)))\right] M_{B,1}}{ 1 + \beta\left[ 1 + \frac{\rho^2}{mB_0K}(\sigma_{\rm max} - \overline{\sigma}(\ell)))\right]} \ - \frac{1}{\varrho} W^*(\overline{\sigma}(\ell)).
\end{equation}
It is clear from \eqref{eq:0502-eq4} and \eqref{eq:0502-eq2} that ${\cal C}_{\rm TM}$ necessarily lies between the two curves ${\cal C}_0$ and ${\cal C}_1$ as shown in Figure \ref{fig:LM1b}.
The monotonicity of ${\cal C}_{\rm TM}$ depicted in the figure follows from simple calculations, the properties of $W^*$ and $\overline{\sigma}$ and the positivity of $M_{B,1} - M_{B,0}$.
By setting $\ell=0$ in \eqref{eq:0502-eq2} we see that ${\cal C}_{\rm TM}$ cuts the $M_1$-axis at
\begin{equation}
\label{eq:0503-eq1}
M_1 = \overline{M}_1(0) = \frac{M_{B,0} + \beta M_{B,1}}{ 1 + \beta} \ - W^*({\sigma}_{\rm max})/\varrho.
\end{equation}
It follows that, corresponding to any given value of the chemical potential
\begin{equation}
\label{eq:0502-eq1}
M_1 > \overline{M}_1(0) = \frac{M_{B,0} + \beta M_{B,1}}{ 1 + \beta} \ - W^*({\sigma}_{\rm max})/\varrho,
\end{equation}
there is a corresponding length of the bar $\ell = \ell_{\rm TM} > 0$ such that $(\ell_{\rm TM}, M_1)$ lies on ${\cal C}_{\rm TM}$, or stated differently, a unique treadmilling solution exists whenever \eqref{eq:0502-eq1} holds. It is not difficult to show that the inequality \eqref{eq:0502-eq1} is equivalent to the inequality \eqref{eq:sb<smax} (which in turn is equivalent to \eqref{eq:R1>R0}).

The value of $\dot\ell$ corresponding to each point $(\ell, M_1)$ of the $\ell, M_1$-plane can be calculated using \eqref{eq:lsdae1} and $\sigma= \overline\sigma(\ell)$ from \eqref{eq:sl}. The arrows in Figure \ref{fig:LM1b} indicate whether this value of $\dot\ell$ is positive or negative. If $M_1$ is held constant (as we have done) then the system starting out initially at some point of the  $\ell, M_1$-plane will follow the $M_1={\rm constant}$ line through that point in the direction indicated by the arrows. The stability of the treadmilling solution is evident.

If the chemical potential of the free monomers does not obey \eqref{eq:0502-eq1}, i.e. if
\begin{equation}
M_1 < \overline{M}_1(0) = \frac{M_{B,0} + \beta M_{B,1}}{ 1 + \beta} \ - W^*({\sigma}_{\rm max})/\varrho,
\end{equation}
we see from Figure \ref{fig:LM1b} that $(a)$ there is no corresponding treadmilling state, and $(b)$ if the bar has some positive length at the initial instant, it will mononotically get shorter and eventually disappear.

Now consider the case $M_{B,1} < M_{B,0}$.
Figure \ref{fig:LM1c} shows the $\ell, M_1$-plane and the signs of the growth velocities on different regions of it. The treadmilling curve, still characterized by \eqref{eq:0502-eq2}, has the properties that $(a)$ it lies between ${\cal C}_0$ and ${\cal C}_1$, $(b)$ it passes through $(0, \overline{M}_1(0))$ where $\overline{M}_1(0)$ continues to be given by
\eqref{eq:0503-eq1}, and $(c)$ $\overline{M}_1(\ell) \to \infty$ as $\ell \to \infty$.  Thus it is clear that there necessarily exists a treadmilling solution corresponding to any $M_1 > \overline{M}_1(0)$ which illustrates the result in Proposition \ref{prop-2}. However the curve ${\cal C}_{\rm TM}$ is not necessarily monotonically rising and so $(i)$ there may exist  multiple treadmiling states corresponding to  a given value of $M_1$ and $(ii)$ there may exist treadmiling states for values of $M_1$ in the range $M_{B,1} - W^*(\sigma_{\rm max})/\varrho < M_1 < \overline{M}_1(0)$.  Note from Figure \ref{fig:LM1c} that, since the treadmilling curve lies between ${\cal C}_0$ and ${\cal C}_1$, all such solutions involve accretion at the right-hand end and ablation at the left-hand in contrast to the case $M_{B,1} > M_{B,0}$ shown in Figure \ref{fig:LM1a}.

\begin{figure}[h]
\centerline{\includegraphics[width=0.75\linewidth]{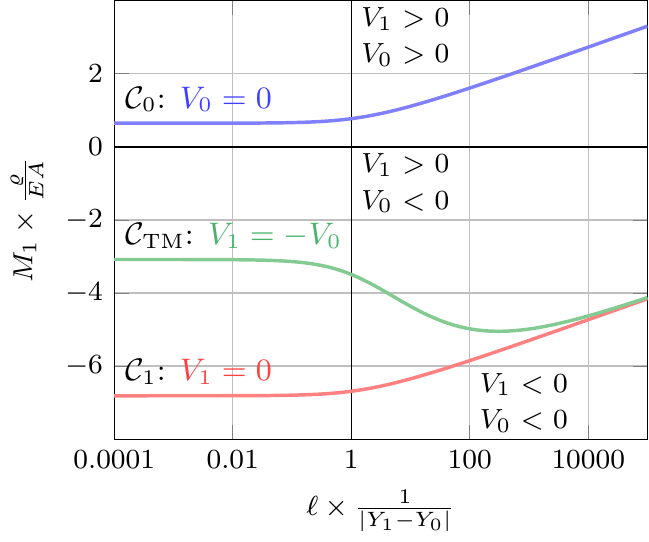}}

\caption{
The $\ell, M_1$-plane in the case $M_{B,1} < M_{B,0}$. The growth rates $V_0$ and $V_1$ vanish on the respective curves ${\cal C}_0$ and ${\cal C}_1$. $M_{B,1} = 0$. Other parameters as in Figure~\ref{fig:tm2}.}

\label{fig:LM1c}
\end{figure}


\subsection{\label{ss:fv} Approach to treadmilling: force-velocity curve}


While  the focus in Section \ref{sec:TM}  was on treadmilling states in which the length of the bar is stationary, the mathematical model constructed in Section \ref{sec:DAE}  can also describe the evolution of the system. For instance it is possible to follow the growth velocity in terms of the applied force as the system tends towards treadmilling. This force - velocity relationship is experimentally measured \citep{ParekhChaudhuri2005,BielingLi2016} and therefore relevant for model validation and discussion.

To calculate the growth velocity, we consider the positions of the two tips of the bar in physical space; they are $y_1(t)$ and $Y_0$, the latter being fixed.  The velocity of the moving tip is therefore $\dot{y}_1$. By differentiating $y_1(t) - Y_0 = \widehat{\lambda}(\sigma(t)) \ell(t)$ with respect to $t$ and using equations \eqref{eq:lsdae1} and \eqref{eq:lsdae2}, one can show that
\begin{equation}\label{eq:y1dot}
  \dot{y}_1 =\frac{\widehat\lambda^2(\sigma)}{\widehat\lambda(\sigma) + (\sigma_{\rm max} - \sigma) \widehat\lambda'(\sigma)} \, [R_1(\sigma) - R_0(\sigma)].
\end{equation}
We are now in the position to simulate an experiment.  For simplicity we limit attention to the case $Y_1=Y_0$ which means that when the spring is force-free, its free end is at the impermeable wall, and so by eq. \eqref{eq:smaxdef}, $\sigma_{\rm max} = 0$. Suppose that initially the bar has zero length $\ell(0)=0$ and the force in the spring vanishes, $\sigma(0)=0$. Thus $y_1(0) = Y_1 = Y_0$.  Equation~\eqref{eq:y1dot} is plotted in Figure \ref{fig:fv} displaying the evolution of the growth velocity  $\dot{y}_1$ as a function of the  compressive force applied by the spring $-\sigma$.  
\begin{figure}[h]
\begin{center}
\includegraphics[width=0.75\linewidth]{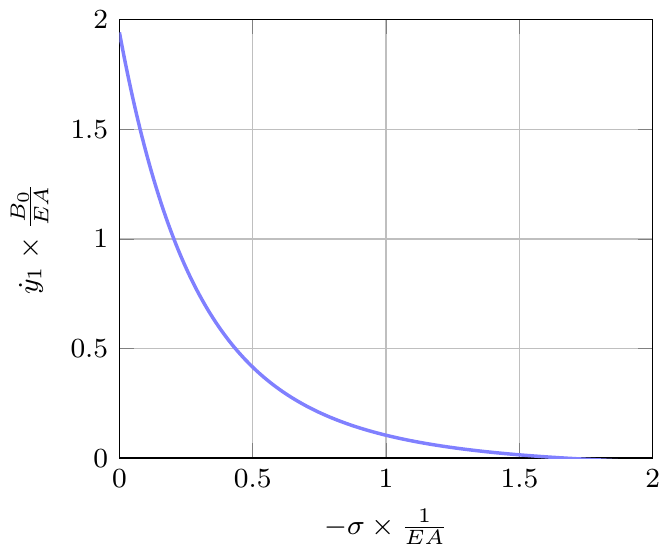}
\end{center}
\caption{\footnotesize Example of evolution of the growth velocity of the bar approaching treadmillinmg. $W^*(\sigma)$ from eq.~\eqref{eq:W*ex2} with parameters  $\sigma_{\rm asymp} = 2 EA$, $\sigma_{\rm max} = 0 EA$, $\sigma_{\alpha 0} = -  4 EA$, $\sigma_{\alpha 1} = -1 EA$, $\beta = 1.0$.}
\label{fig:fv}
\normalsize
\end{figure}
At the beginning, the force $\sigma$ is zero and the velocity is maximum. Then, as the bar grows, the spring is shortened and the compressive force in the bar increases. In turn, compression affects the growth velocity which progressively decreases, until it eventually reaches treadmilling and stops.

The curve in Figure~\ref{fig:fv} is convex. Experimental results in \cite{ParekhChaudhuri2005} in which the elastic restraint is represented by the cantilever tip of an atomic force microscope exhibit instead a concave curve which is approximately horizontal for low values of the force. The authors in \cite{ParekhChaudhuri2005} provide two possibile explanations for their results:  first, the remodeling and increased density of actin networks growing under increasing loads. Second, force-independent limitations to growth ``such as the inherent time required for polymerization of actin or nucleation of new filaments''. Both effects are currently not accounted in the present model.

A more recent paper by the same group obtains an experimental force-velocity curve \citep[see Figure 2B in][]{BielingLi2016} that is convex and qualitatively similar to the one in Figure~\ref{fig:fv}. However the curve in \cite{BielingLi2016} is obtained by measuring the steady-state growth velocity of an actin network subject to a step-wise series of increasing load forces. Therefore it does not reflect the evolution of a specimen towards treadmilling but the sampling of growth velocities of a specimen under different constant values of force.

Among other possible differences between the present model and the experimental setup in \cite{ParekhChaudhuri2005} and \cite{BielingLi2016} is the fact that in the experiments free monomers can diffuse from the sides of the specimen while in the present model actin monomers flow longitudinally along the bar whose length evolves in time.

Thanks to the many experimental results in the literature pertaining to the relatively simple one-dimensional context, e.g. \cite{ParekhChaudhuri2005} and \cite{BielingLi2016}, we plan to carry out appropriate refinements of the present work in the future.



\section{\label{sec:end}Conclusions}
In this study we have formulated and analyzed a one-dimensional, self-contained growth model for an actin bar fixed at one end and elastically constrained at the other. The model encompasses diffusion in a solvent surrounding or permeating the bar as well as growth conditions at the ends. The nonlinear elastic properties of actin are specified through an arbitrary convex strain energy density function.

Treadmilling states were investigated in which the length of the bar remains constant while accreting actin monomers at one end and ablating them at the opposite end at equal rates. The treadmilling state in which monomers accrete at the fixed end is found to be globally stable. Multiple treadmilling states in which monomers ablate at the fixed end are instead possible.

Conditions for existence and stability of treadmilling states were condensed into  relatively simple formulas, \eqref{eq:TMstable}, \eqref{eq:R1>R0} and \eqref{eq:sb<smax}, which are useful in understanding the influence of the different parameters of the model, and may help explain the results of more complex models. For instance, global stability of treamilling states accreting monomers at the fixed end can provide additional explanation to the numerical findings made by \cite{Abi-AklAbeyaratne2019} in a similar setting.

Further refinements of the model would be relevant for applications, especially experimental ones: accounting for the density and stiffness increase of actin under higher external load and developing analytical relationships between growth velocity and applied force are two of them. Finally, extension of the present work to the stability of two dimensional treadmilling states previously studied in \cite{TomassettiCohen:2016} is planned.

\section*{Acknowledgements}
R.A. and E.P. gratefully acknowledge the support of the MIT-FVG Seed Fund. The work of G.T. and E.P. was supported by the Italian PRIN 2017 project ``Mathematics of active materials: From mechanobiology to smart devices''.  E.P. thankfully acknowledges as well the support of the Italian National Group of Mathematical Physics (GNFM-INdAM). G.T. acknowledges the Grant of Excellence Departments, MIUR-Italy (Art.$1$, commi $314$-$337$, Legge $232$/$2016$).

\appendix
\section{\label{app:Wex} Examples of strain energy densities}

We provide here a couple of examples of strain energy density functions satisfying the requirements in eq.~\eqref{eq:Wass}. We add also some remarks on the asymptotic growth of the constitutive laws $\widehat{\sigma}(\lambda)$ for stress and $W^*(\sigma)$ for the complementary strain energy density function ensuing from the choices of the asymptotic dominant term of the strain energy density $W$.

\subsubsection*{Example 1: A rational function}

First consider the rational strain energy density
  \begin{equation} \label{eq:Wex1}
    W(\lambda) = \frac {EA} 6( \lambda^2 + 2 \lambda^{-1} -3) \, ,
  \end{equation}
where $EA>0$ is a constant. It follows that $\sigma = W'(\lambda) = \frac{EA}{3}(\lambda - \lambda^{-2}),  W''(\lambda) = \tfrac{EA}{3}(1 + 2 \lambda^{-3}) >0$ and $W''(1) = EA$. Notice that this $W$ obeys all assumptions in \eqref{eq:Wass}. Then the complementary energy $W^*$ can be obtained,
  \begin{equation*}
    W^* = \frac{EA}{6} \big[\lambda^2 - 4 \lambda^{-1} + 3\big]\ \Big|_{\lambda = \widehat{\lambda}(\sigma)} \,,
  \end{equation*}
  together with its derivatives,
  \begin{equation*}
    W^*\, '(\sigma) = \lambda \Big|_{\lambda = \widehat{\lambda}(\sigma)} >0
    \quad \text{and} \quad
    W^*\, ''(\sigma) = \frac{3}{EA(1+ 2\lambda^{-3})} \Big|_{\lambda = \widehat{\lambda}(\sigma)} >0 \,.
  \end{equation*}
  The closed form expression of $W^*$ is cumbersome, but we can easily look at its asymptotic behaviour when $\sigma \to \pm \infty$:
  \begin{equation*}
    W^* \sim \frac 32 \frac{\sigma^2}{EA}  \to +\infty \quad {\rm for } \ \sigma \to +\infty, \qquad
  W^* \sim  -\, \sqrt{ \frac 43 EA} \ \sqrt{-\sigma}  \to -\infty \quad {\rm for } \  \sigma \to - \infty
  \end{equation*}

  More generally, suppose that
  \begin{equation*}
    W(\lambda) \sim \alpha \lambda^n, \qquad \sigma \sim \alpha n \lambda^{n-1},
    \qquad \alpha >0, n >1, \qquad {\rm as} \  \lambda \to \infty.
  \end{equation*}
  Then
  \begin{equation*}
    \lambda \sim \left(\frac{\sigma}{\alpha n}\right)^{\frac{1}{n-1}}, \qquad W^* \sim \alpha (n-1)
    \left(\frac{\sigma}{\alpha n}\right)^{\frac{n}{n-1}} \to +\infty \qquad {\rm as} \ \sigma \to +\infty.
  \end{equation*}
  Similarly suppose
  \begin{equation*}
    W(\lambda) \sim \beta \lambda^{-m}, \qquad \sigma \sim - \beta m \lambda^{-m-1},  \qquad \beta >0, m > 0, \qquad {\rm as} \ \lambda \to 0.
  \end{equation*}
  Then
  \begin{equation*}
    \lambda \sim \left(\frac{\beta m}{-\sigma}\right)^{\frac{1}{m+1}}, \qquad W^* \sim - \beta (m+1)
    \left(\frac{-\sigma}{\beta m}\right)^{\frac{m}{m+1}} \to - \infty \qquad {\rm as} \ \sigma \to -\infty
  \end{equation*}

  \subsubsection*{Example 2: A closed form expression}
  A second example of a suitable strain energy density is
  \begin{equation} \label{eq:Wex2}
    W(\lambda) = \frac{EA}{2} \left( \frac{\lambda^2}{2} - \ln \lambda - \frac 12\right) \,, \qquad
    \sigma = W'(\lambda) = \frac{EA}{2}\left( \lambda - \lambda^{-1}\right) \;,
  \end{equation}
 where $EA>0$ is a constant. Associated with this $W$ one has, in closed form,
  \begin{equation} \label{eq:W*ex2}
    \widehat{\lambda}(\sigma) = \frac{\sigma}{EA} + \sqrt{ \frac{\sigma^2}{(EA)^2} + 1} \,; \quad W^*(\sigma) = \frac{EA}{2} \left( \frac{\widehat{\lambda}^2(\sigma)}{2} + \ln \widehat{\lambda}(\sigma) - \frac 12\right) \,.
  \end{equation}
  Also in this case we have that $W''(1) = EA$ and one can readily verify that $W$ satisfies all requirements in \eqref{eq:Wass}. Plots of $W(\lambda)$ and $W^*(\sigma)$ are shown in Figure~\ref{fig:W}.
  \begin{figure}[h]
  \begin{center}
  \includegraphics[width=0.9\linewidth]{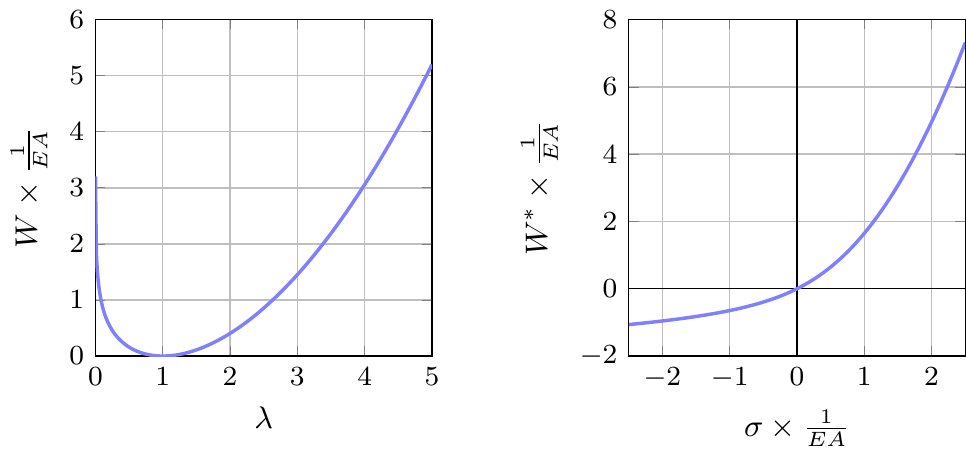}
  \end{center}
  \caption{\footnotesize Plot of $W(\lambda)$ in \eqref{eq:Wex2} and of $W^*(\sigma)$ in \eqref{eq:W*ex2}.}
  \label{fig:W}
  \normalsize
  \end{figure}

\noindent

\bibliographystyle{elsarticle-harv}

\bibliography{1dgrowth}

\begin{thebibliography}{29}
\expandafter\ifx\csname natexlab\endcsname\relax\def\natexlab#1{#1}\fi
\providecommand{\url}[1]{\texttt{#1}}
\providecommand{\href}[2]{#2}
\providecommand{\path}[1]{#1}
\providecommand{\DOIprefix}{doi:}
\providecommand{\ArXivprefix}{arXiv:}
\providecommand{\URLprefix}{URL: }
\providecommand{\Pubmedprefix}{pmid:}
\providecommand{\doi}[1]{\href{http://dx.doi.org/#1}{\path{#1}}}
\providecommand{\Pubmed}[1]{\href{pmid:#1}{\path{#1}}}
\providecommand{\bibinfo}[2]{#2}
\ifx\xfnm\relax \def\xfnm[#1]{\unskip,\space#1}\fi
\bibitem[{Abeyaratne and Knowles(1990)}]{AbeyaratneKnowles1990}
\bibinfo{author}{Abeyaratne, R.}, \bibinfo{author}{Knowles, J.K.},
  \bibinfo{year}{1990}.
\newblock \bibinfo{title}{On the driving traction acting on a surface of strain
  discontinuity in a continuum}.
\newblock \bibinfo{journal}{Journal of the Mechanics and Physics of Solids}
  \bibinfo{volume}{38}, \bibinfo{pages}{345--360}.
\newblock \DOIprefix\doi{10.1016/0022-5096(90)90003-M}.
\bibitem[{Abeyaratne and Knowles(1997)}]{AbeyaratneKnowles1997}
\bibinfo{author}{Abeyaratne, R.}, \bibinfo{author}{Knowles, J.K.},
  \bibinfo{year}{1997}.
\newblock \bibinfo{title}{A note on the driving traction acting on a
  propagating interface: Adiabatic and non-adiabatic processes of a continuum}.
\newblock \bibinfo{journal}{Journal of Applied Mechanics} \bibinfo{volume}{67},
  \bibinfo{pages}{829--830}.
\newblock \DOIprefix\doi{10.1115/1.1308577}.
\bibitem[{Abi-Akl et~al.(2019)Abi-Akl, Abeyaratne and
  Cohen}]{Abi-AklAbeyaratne2019}
\bibinfo{author}{Abi-Akl, R.}, \bibinfo{author}{Abeyaratne, R.},
  \bibinfo{author}{Cohen, T.}, \bibinfo{year}{2019}.
\newblock \bibinfo{title}{Kinetics of surface growth with coupled diffusion and
  the emergence of a universal growth path}.
\newblock \bibinfo{journal}{Proceedings of the Royal Society A: Mathematical,
  Physical and Engineering Sciences} \bibinfo{volume}{475},
  \bibinfo{pages}{20180465}.
\newblock \DOIprefix\doi{10.1098/rspa.2018.0465}.
\bibitem[{Alberts et~al.(2015)Alberts, Johnson, Lewis, Morgan and
  Raff}]{AlbertsJohnson2015}
\bibinfo{author}{Alberts, B.}, \bibinfo{author}{Johnson, A.},
  \bibinfo{author}{Lewis, J.}, \bibinfo{author}{Morgan, D.},
  \bibinfo{author}{Raff, M.}, \bibinfo{year}{2015}.
\newblock \bibinfo{title}{Molecular Biology of the Cell}.
\newblock \bibinfo{edition}{6} ed., \bibinfo{publisher}{Garland Science}.
\bibitem[{Bacigalupo and Gambarotta(2012)}]{BacigalupoGambarotta:2012}
\bibinfo{author}{Bacigalupo, A.}, \bibinfo{author}{Gambarotta, L.},
  \bibinfo{year}{2012}.
\newblock \bibinfo{title}{Effects of layered accretion on the mechanics of
  masonry structures}.
\newblock \bibinfo{journal}{Mechanics Based Design of Structures and Machines}
  \bibinfo{volume}{40}, \bibinfo{pages}{163--184}.
\newblock \DOIprefix\doi{10.1080/15397734.2011.628622}.
\bibitem[{Bieling et~al.(2016)Bieling, Li, Weichsel, McGorty, Jreij, Huang,
  Fletcher and Mullins}]{BielingLi2016}
\bibinfo{author}{Bieling, P.}, \bibinfo{author}{Li, T.D.},
  \bibinfo{author}{Weichsel, J.}, \bibinfo{author}{McGorty, R.},
  \bibinfo{author}{Jreij, P.}, \bibinfo{author}{Huang, B.},
  \bibinfo{author}{Fletcher, D.}, \bibinfo{author}{Mullins, R.D.},
  \bibinfo{year}{2016}.
\newblock \bibinfo{title}{Force feedback controls motor activity and mechanical
  properties of self-assembling branched actin networks}.
\newblock \bibinfo{journal}{Cell} \bibinfo{volume}{164},
  \bibinfo{pages}{115--127}.
\newblock \DOIprefix\doi{10.1016/j.cell.2015.11.057}.
\bibitem[{Bindschadler et~al.(2004)Bindschadler, Osborn, Dewey and
  McGrath}]{BindschadlerOsborn2004}
\bibinfo{author}{Bindschadler, M.}, \bibinfo{author}{Osborn, E.A.},
  \bibinfo{author}{Dewey, C.F.}, \bibinfo{author}{McGrath, J.L.},
  \bibinfo{year}{2004}.
\newblock \bibinfo{title}{A mechanistic model of the actin cycle}.
\newblock \bibinfo{journal}{Biophysical Journal} \bibinfo{volume}{86},
  \bibinfo{pages}{2720--2739}.
\newblock \DOIprefix\doi{10.1016/S0006-3495(04)74326-X}.
\bibitem[{de~Buyl et~al.(2013)de~Buyl, Mikhailov and
  Kapral}]{BuylMikhailov2013}
\bibinfo{author}{de~Buyl, P.}, \bibinfo{author}{Mikhailov, A.S.},
  \bibinfo{author}{Kapral, R.}, \bibinfo{year}{2013}.
\newblock \bibinfo{title}{Self-propulsion through symmetry breaking}.
\newblock \bibinfo{journal}{EPL (Europhysics Letters)} \bibinfo{volume}{103},
  \bibinfo{pages}{60009}.
\newblock \DOIprefix\doi{10.1209/0295-5075/103/60009}.
\bibitem[{Cameron et~al.(1999)Cameron, Footer, van Oudenaarden and
  Theriot}]{CameronFooter1999}
\bibinfo{author}{Cameron, L.A.}, \bibinfo{author}{Footer, M.J.},
  \bibinfo{author}{van Oudenaarden, A.}, \bibinfo{author}{Theriot, J.A.},
  \bibinfo{year}{1999}.
\newblock \bibinfo{title}{Motility of acta protein-coated microspheres driven
  by actin polymerization}.
\newblock \bibinfo{journal}{Proc Natl Acad Sci USA} \bibinfo{volume}{96},
  \bibinfo{pages}{4908--4913}.
\newblock \DOIprefix\doi{10.1073/pnas.96.9.4908}.
\bibitem[{Cardamone et~al.(2011)Cardamone, Laio, Torre, Shahapure and
  DeSimone}]{CardamoneLaio2011}
\bibinfo{author}{Cardamone, L.}, \bibinfo{author}{Laio, A.},
  \bibinfo{author}{Torre, V.}, \bibinfo{author}{Shahapure, R.},
  \bibinfo{author}{DeSimone, A.}, \bibinfo{year}{2011}.
\newblock \bibinfo{title}{Cytoskeletal actin networks in motile cells are
  critically self-organized systems synchronized by mechanical interactions}.
\newblock \bibinfo{journal}{Proc Natl Acad Sci USA} \bibinfo{volume}{108},
  \bibinfo{pages}{13978--13983}.
\newblock \DOIprefix\doi{10.1073/pnas.1100549108}.
\bibitem[{Chaudhuri et~al.(2007)Chaudhuri, Parekh and
  Fletcher}]{ChaudhuriParekh2007}
\bibinfo{author}{Chaudhuri, O.}, \bibinfo{author}{Parekh, S.H.},
  \bibinfo{author}{Fletcher, D.A.}, \bibinfo{year}{2007}.
\newblock \bibinfo{title}{Reversible stress softening of actin networks}.
\newblock \bibinfo{journal}{Nature} \bibinfo{volume}{445},
  \bibinfo{pages}{295--298}.
\newblock \DOIprefix\doi{10.1038/nature05459}.
\bibitem[{Ciarletta et~al.(2013)Ciarletta, Preziosi and
  Maugin}]{CiarlettaPreziosi2013}
\bibinfo{author}{Ciarletta, P.}, \bibinfo{author}{Preziosi, L.},
  \bibinfo{author}{Maugin, G.A.}, \bibinfo{year}{2013}.
\newblock \bibinfo{title}{Mechanobiology of interfacial growth}.
\newblock \bibinfo{journal}{Journal of the Mechanics and Physics of Solids}
  \bibinfo{volume}{61}, \bibinfo{pages}{852--872}.
\newblock \DOIprefix\doi{10.1016/j.jmps.2012.10.011}.
\bibitem[{Edelstein-Keshet and
  Ermentrout(2000)}]{Edelstein-KeshetErmentrout2000}
\bibinfo{author}{Edelstein-Keshet, L.}, \bibinfo{author}{Ermentrout, G.B.},
  \bibinfo{year}{2000}.
\newblock \bibinfo{title}{Models for spatial polymerization dynamics of
  rod-like polymers}.
\newblock \bibinfo{journal}{Journal of Mathematical Biology}
  \bibinfo{volume}{40}, \bibinfo{pages}{64--96}.
\newblock \DOIprefix\doi{10.1007/s002850050005}.
\bibitem[{Ganghoffer and Goda(2018)}]{GanghofferGoda2018}
\bibinfo{author}{Ganghoffer, J.F.}, \bibinfo{author}{Goda, I.},
  \bibinfo{year}{2018}.
\newblock \bibinfo{title}{A combined accretion and surface growth model in the
  framework of irreversible thermodynamics}.
\newblock \bibinfo{journal}{International Journal of Engineering Science}
  \bibinfo{volume}{127}, \bibinfo{pages}{53--79}.
\newblock \DOIprefix\doi{10.1016/j.ijengsci.2018.02.006}.
\bibitem[{Goriely(2017)}]{Goriely:2017}
\bibinfo{author}{Goriely, A.}, \bibinfo{year}{2017}.
\newblock \bibinfo{title}{The Mathematics and Mechanics of Biological Growth}.
  volume~\bibinfo{volume}{45} of \textit{\bibinfo{series}{Interdisciplinary
  Applied Mathematics}}.
\newblock \bibinfo{edition}{1} ed., \bibinfo{publisher}{Springer}.
\newblock \DOIprefix\doi{10.1007/978-0-387-87710-5}.
\bibitem[{van~der Gucht et~al.(2005)van~der Gucht, Paluch, Plastino and
  Sykes}]{GuchtPaluch2005}
\bibinfo{author}{van~der Gucht, J.}, \bibinfo{author}{Paluch, E.},
  \bibinfo{author}{Plastino, J.}, \bibinfo{author}{Sykes, C.},
  \bibinfo{year}{2005}.
\newblock \bibinfo{title}{Stress release drives symmetry breaking for
  actin-based movement}.
\newblock \bibinfo{journal}{Proc Natl Acad Sci USA} \bibinfo{volume}{102},
  \bibinfo{pages}{7847}.
\newblock \DOIprefix\doi{10.1073/pnas.0502121102}.
\bibitem[{John et~al.(2008)John, Peyla, Kassner, Prost and
  Misbah}]{JohnPeyla2008}
\bibinfo{author}{John, K.}, \bibinfo{author}{Peyla, P.},
  \bibinfo{author}{Kassner, K.}, \bibinfo{author}{Prost, J.},
  \bibinfo{author}{Misbah, C.}, \bibinfo{year}{2008}.
\newblock \bibinfo{title}{Nonlinear study of symmetry breaking in actin gels:
  Implications for cellular motility}.
\newblock \bibinfo{journal}{Physical Review Letters} \bibinfo{volume}{100},
  \bibinfo{pages}{068101}.
\newblock \DOIprefix\doi{10.1103/PhysRevLett.100.068101}.
\bibitem[{Noireaux et~al.(2000)Noireaux, Golsteyn, Friederich, Prost, Antony,
  Louvard and Sykes}]{NoireauxGolsteyn:2000}
\bibinfo{author}{Noireaux, V.}, \bibinfo{author}{Golsteyn, R.M.},
  \bibinfo{author}{Friederich, E.}, \bibinfo{author}{Prost, J.},
  \bibinfo{author}{Antony, C.}, \bibinfo{author}{Louvard, D.},
  \bibinfo{author}{Sykes, C.}, \bibinfo{year}{2000}.
\newblock \bibinfo{title}{Growing an actin gel on spherical surfaces}.
\newblock \bibinfo{journal}{Biophysical Journal} \bibinfo{volume}{78},
  \bibinfo{pages}{1643--1654}.
\newblock \DOIprefix\doi{10.1016/S0006-3495(00)76716-6}.
\bibitem[{Parekh et~al.(2005)Parekh, Chaudhuri, Theriot and
  Fletcher}]{ParekhChaudhuri2005}
\bibinfo{author}{Parekh, S.H.}, \bibinfo{author}{Chaudhuri, O.},
  \bibinfo{author}{Theriot, J.A.}, \bibinfo{author}{Fletcher, D.A.},
  \bibinfo{year}{2005}.
\newblock \bibinfo{title}{Loading history determines the velocity of
  actin-network growth}.
\newblock \bibinfo{journal}{Nature Cell Biology} \bibinfo{volume}{7},
  \bibinfo{pages}{1219--1223}.
\newblock \DOIprefix\doi{10.1038/ncb1336}.
\bibitem[{Prost et~al.(2015)Prost, Jülicher and Joanny}]{ProstJuelicher2015}
\bibinfo{author}{Prost, J.}, \bibinfo{author}{Jülicher, F.},
  \bibinfo{author}{Joanny, J.F.}, \bibinfo{year}{2015}.
\newblock \bibinfo{title}{Active gel physics}.
\newblock \bibinfo{journal}{Nature Physics} \bibinfo{volume}{11},
  \bibinfo{pages}{111--117}.
\newblock \DOIprefix\doi{10.1038/nphys3224}.
\bibitem[{Prost et~al.(2008)Prost, Joanny, Lenz and Sykes}]{ProstJoanny2008}
\bibinfo{author}{Prost, J.}, \bibinfo{author}{Joanny, J.F.},
  \bibinfo{author}{Lenz, P.}, \bibinfo{author}{Sykes, C.},
  \bibinfo{year}{2008}.
\newblock \bibinfo{title}{The physics of listeria propulsion}, in:
  \bibinfo{editor}{Lenz, P.} (Ed.), \bibinfo{booktitle}{Cell Motility}.
  \bibinfo{publisher}{Springer New York}, \bibinfo{address}{New York, NY}.
  Biological and Medical Physics, Biomedical Engineering.
  chapter~\bibinfo{chapter}{1}, pp. \bibinfo{pages}{1--30}.
\newblock \DOIprefix\doi{10.1007/978-0-387-73050-9_1}.
\bibitem[{Skalak et~al.(1982)Skalak, Dasgupta, Moss, Otten, Dullemeijer and
  Vilmann}]{SkalakDasgupta1982}
\bibinfo{author}{Skalak, R.}, \bibinfo{author}{Dasgupta, G.},
  \bibinfo{author}{Moss, M.}, \bibinfo{author}{Otten, E.},
  \bibinfo{author}{Dullemeijer, P.}, \bibinfo{author}{Vilmann, H.},
  \bibinfo{year}{1982}.
\newblock \bibinfo{title}{Analytical description of growth}.
\newblock \bibinfo{journal}{Journal of Theoretical Biology}
  \bibinfo{volume}{94}, \bibinfo{pages}{555--577}.
\newblock \DOIprefix\doi{10.1016/0022-5193(82)90301-0}.
\bibitem[{Skalak et~al.(1997)Skalak, Farrow and Hoger}]{SkalakFarrow1997}
\bibinfo{author}{Skalak, R.}, \bibinfo{author}{Farrow, D.A.},
  \bibinfo{author}{Hoger, A.}, \bibinfo{year}{1997}.
\newblock \bibinfo{title}{Kinematics of surface growth}.
\newblock \bibinfo{journal}{Journal of Mathematical Biology}
  \bibinfo{volume}{35}, \bibinfo{pages}{869--907}.
\newblock \DOIprefix\doi{10.1007/s002850050081}.
\bibitem[{Theriot(2000)}]{Theriot2000}
\bibinfo{author}{Theriot, J.A.}, \bibinfo{year}{2000}.
\newblock \bibinfo{title}{The polymerization motor}.
\newblock \bibinfo{journal}{Traffic} \bibinfo{volume}{1},
  \bibinfo{pages}{19--28}.
\newblock \DOIprefix\doi{10.1034/j.1600-0854.2000.010104.x}.
\bibitem[{Tomassetti et~al.(2016)Tomassetti, Cohen and
  Abeyaratne}]{TomassettiCohen:2016}
\bibinfo{author}{Tomassetti, G.}, \bibinfo{author}{Cohen, T.},
  \bibinfo{author}{Abeyaratne, R.}, \bibinfo{year}{2016}.
\newblock \bibinfo{title}{Steady accretion of an elastic body on a hard
  spherical surface and the notion of a four-dimensional reference space}.
\newblock \bibinfo{journal}{Journal of the Mechanics and Physics of Solids}
  \bibinfo{volume}{96}, \bibinfo{pages}{333--352}.
\newblock \DOIprefix\doi{10.1016/j.jmps.2016.05.015}.
\bibitem[{Truskinovsky and Zurlo(2019)}]{TruskinovskyZurlo2019}
\bibinfo{author}{Truskinovsky, L.}, \bibinfo{author}{Zurlo, G.},
  \bibinfo{year}{2019}.
\newblock \bibinfo{title}{Nonlinear elasticity of incompatible surface growth}.
\newblock \bibinfo{journal}{Physical Review E} \bibinfo{volume}{99},
  \bibinfo{pages}{053001}.
\newblock \DOIprefix\doi{10.1103/physreve.99.053001}.
\bibitem[{Zimmermann(2014)}]{Zimmermann2014}
\bibinfo{author}{Zimmermann, J.}, \bibinfo{year}{2014}.
\newblock \bibinfo{title}{Modeling the lamellipodium of motile cells}.
\newblock Ph.D. thesis. Humboldt-Universität zu Berlin,
  Mathematisch-Naturwissenschaftliche Fakultät I.
\newblock \DOIprefix\doi{10.18452/16871}.
\bibitem[{Zurlo and Truskinovsky(2017)}]{ZurloTruskinovsky2017}
\bibinfo{author}{Zurlo, G.}, \bibinfo{author}{Truskinovsky, L.},
  \bibinfo{year}{2017}.
\newblock \bibinfo{title}{Printing non-euclidean solids}.
\newblock \bibinfo{journal}{Physical Review Letters} \bibinfo{volume}{119},
  \bibinfo{pages}{048001}.
\newblock \DOIprefix\doi{10.1103/PhysRevLett.119.048001}.
\bibitem[{Zurlo and Truskinovsky(2018)}]{ZurloTruskinovsky2018}
\bibinfo{author}{Zurlo, G.}, \bibinfo{author}{Truskinovsky, L.},
  \bibinfo{year}{2018}.
\newblock \bibinfo{title}{Inelastic surface growth}.
\newblock \bibinfo{journal}{Mechanics Research Communications}
  \bibinfo{volume}{93}, \bibinfo{pages}{174--179}.
\newblock \DOIprefix\doi{10.1016/j.mechrescom.2018.01.007}.

\end{thebibliography}

\end{document}